\documentclass[12pt,letterpaper]{article}

\def\spacingset#1{\renewcommand{\baselinestretch}%
{#1}\small\normalsize} \spacingset{1}

\usepackage{natbib}
\usepackage[ left=.9in, top=0.9in, right=0.9in, bottom=0.9in]{geometry}
\usepackage{graphicx,bm,bbm,colonequals,amsmath,amssymb,url}
\usepackage{array,tabularx,multirow}
\usepackage{enumitem,algpseudocode}
\usepackage[font={footnotesize}]{caption,subcaption}
\usepackage[usenames, dvipsnames]{color}
\usepackage{cancel}

\bibpunct[, ]{(}{)}{;}{a}{,}{,}

\usepackage{amsthm}

\newtheorem{prop}{Proposition}
\newtheoremstyle{examplestyle}
  {3pt}
  {4pt}
  {\itshape}
  {0pt}
  {\bfseries}
  {.}
  { }
  {\thmname{#1}\thmnumber{ #2}\thmnote{ (#3)}}

\theoremstyle{examplestyle}
\newtheorem{example}{Example}

\usepackage{array,framed}
\newcounter{algorithm}
\newenvironment{algorithm}[1][]{\refstepcounter{algorithm}\par\medskip\noindent%
   \textbf{Algorithm~\thealgorithm: #1} \rmfamily}{\medskip}

\newcommand{\bv}{\mathbf{v}}

\newcommand{\bx}{\mathbf{x}}
\newcommand{\by}{\mathbf{y}}

\newcommand{\bw}{\mathbf{w}}

\newcommand{\bz}{\mathbf{z}}

\newcommand{\bI}{\mathbf{I}}

\newcommand{\bH}{\mathbf{H}}

\newcommand{\bK}{\mathbf{K}}

\newcommand{\bQ}{\mathbf{Q}}

\newcommand{\bC}{\mathbf{C}}
\newcommand{\bM}{\mathbf{M}}
\newcommand{\bR}{\mathbf{R}}

\newcommand{\bfzero}{\mathbf{0}}
\newcommand{\bfone}{\mathbf{1}}

\newcommand{\bfgamma}{\bm{\gamma}}
\newcommand{\bfmu}{\bm{\mu}}

\newcommand{\bftheta}{\bm{\theta}}

\newcommand{\bfnu}{\bm{\nu}}

\newcommand{\bfSigma}{\bm{\Sigma}}

\newcommand{\bfLambda}{\bm{\Lambda}}

\DeclareMathOperator{\var}{var}

\DeclareMathOperator{\diag}{diag}

\newcommand{\evol}{\mathcal{M}}
\newcommand{\levol}{\mathbf{M}}
\renewcommand{\lor}{\textnormal{Lor}}
\newcommand{\matern}{\textnormal{Mat}}
\newcommand{\obs}{\mathcal{H}}
\newcommand{\lobs}{\mathbf{H}}
\newcommand{\normal}{\mathcal{N}}
\newcommand{\ig}{\mathcal{IG}}
\newcommand{\update}{\mathcal{U}}
\newcommand{\likelihood}{\mathcal{L}}
\newcommand{\elik}{\mathcal{L}^{\textnormal{E}}}
\newcommand{\plik}{\mathcal{L}^{\textnormal{P}}}
\newcommand{\mlik}{\mathcal{L}^{\textnormal{Z}}}

\newcommand{\forex}{\mathbf{\widetilde{x}}}

\DeclareMathOperator*{\argmax}{arg\,max}

\def\picdir{plots}

\title{\textcolor{black}{Ensemble Kalman methods for high-dimensional hierarchical dynamic space-time models}}

\author{Matthias Katzfuss\thanks{Texas A\&M University (katzfuss@tamu.edu)} \and Jonathan R.\ Stroud\thanks{Georgetown University} \and Christopher K.\ Wikle\thanks{University of Missouri}}

\date{}

\begin{document}

\maketitle

\begin{abstract}

\textcolor{black}{We propose a new class of filtering and smoothing methods for inference in high-dimensional, nonlinear, non-Gaussian, spatio-temporal state-space models. The main idea is to combine the ensemble Kalman filter and smoother, developed in the geophysics literature, with state-space algorithms from the statistics literature. Our algorithms address a variety of estimation scenarios, including on-line and off-line state and parameter estimation. We take a Bayesian perspective, for which the goal is to generate samples from the joint posterior distribution of states and parameters. The key benefit of our approach is the use of ensemble Kalman methods for dimension reduction, which allows inference for high-dimensional state vectors.  We compare our methods to existing ones, including ensemble Kalman filters, particle filters, and particle MCMC. Using a real data example of cloud motion and data simulated under a number of nonlinear and non-Gaussian scenarios, we show that our approaches outperform these existing methods.}

\end{abstract}

\noindent%
{\it Keywords:} data assimilation; geoscience applications; Gibbs sampler; particle filter; spatio-temporal statistics; state-space models
\vfill

\newpage

\section{Introduction \label{sec:intro}}

\subsection{Inference in state-space models}

\textcolor{black}{
State-space models (SSMs) are general representations of systems observed over time. A SSM describes the temporal evolution of the system state, and the relationship of the state to observations \textcolor{black}{\citep[e.g.,][]{WestHarr:97,Shum:Stof:06}}.  Here, we analyze spatio-temporal SSMs describing the evolution of one or more discretized spatial fields over time.  These models arise, for example, in geoscience and biomedical applications.  They are often nonlinear, non-Gaussian, and have very high-dimensional states.  Standard statistical approaches such as Markov chain Monte Carlo (MCMC) and particle filters break down for this class of models.  Thus, new methods are needed for accurate inference.  We propose a number of new filtering and smoothing algorithms for Bayesian inference in these high-dimensional SSMs.}


\textcolor{black}{
There is an enormous statistics and engineering literature on state and parameter estimation for low-dimensional SSMs.  For linear, Gaussian models, the Kalman filter \citep{Kalman1960} provides closed-form posterior distributions for the states conditional on the parameters. For conditionally Gaussian models, inference is often carried out using \emph{MCMC methods} based on the forward-filter backward-sampling algorithm \citep{Carter1994,Fruhwirth1994} or the mixture Kalman filter \citep{Chen2000}.  MCMC methods have also been developed for more general SSMs, including nonlinear, non-Gaussian models \citep{Carlin1992} and non-Gaussian measurement models \citep{Shephard1997,Gamerman1998}. However, these methods do not scale to high-dimensional states because their computational cost is cubic in the state dimension.}



\textcolor{black}{
Another popular approach for nonlinear, non-Gaussian SSMs is sequential Monte Carlo (SMC) methods, also known as particle filters \cite[e.g.,][]{Doucet2001}.  SMC methods approximate the state distribution at each time by a weighted set of samples or {\em particles}.  These particles are propagated forward through time according to the state evolution.  They are then updated based on new data using a reweighting or resampling step \textcolor{black}{\citep[e.g.,][]{Gordon1993,Liu1998}}.  Although these methods are exact with an infinite number of particles, it is well known that they suffer from weight degeneracy (i.e., all weights but one become essentially zero) in high-dimensional problems \citep[e.g.,][]{Snyd:Beng:Bick:08}.  Thus, particle filters do not scale to high dimensions.}



\textcolor{black}{
In contrast to the statistical approaches above, approximate state-estimation methods in the geoscience literature scale well to high-dimensional problems.  In this paper, we focus on a geophysical data assimilation method called the \emph{ensemble Kalman filter} \citep[EnKF;][]{Even:94}; see \cite{Houtekamer2016} and \cite{Katzfuss2015b} for recent reviews, and \citet{Wikle2007} for a general introduction to data assimilation. The EnKF also uses an ensemble representation that is propagated forward like the particle filter, but instead of reweighting, the EnKF updates the ensemble using a linear ``shift'' that approximates the best linear update. While the EnKF is based on the assumption of a linear, Gaussian model, it provides accurate inference for many nonlinear models as well.  This has allowed the EnKF to be successfully applied to nonlinear spatio-temporal SSMs in millions of dimensions \citep[e.g.,][]{Szunyogh2008,Whitaker2008,Houtekamer2014}. 
}

\textcolor{black}{
The EnKF is also used for parameter estimation and non-Gaussian data assimilation. The most common EnKF method for parameter estimation is {\em state augmentation} \citep{Ande:01}, which refers to simply including the parameters in the state vector.  While EnKF with state augmentation often works well for certain parameters, such as autoregressive coefficients, the resulting linear update can be highly problematic for several types of parameters, most notably, variance and covariance parameters \citep{Stro:Beng:07,DelSole2010}. The linear shifting update in the EnKF can also be highly suboptimal for non-Gaussian observations.  Alternative non-Gaussian EnKF approaches, including quantile-based filters \citep{Anderson2010} and moment-matching filters \citep{Lei:Bick:11}, still perform poorly for certain classes of measurement distributions.   Thus, new EnKF methods are needed for more general state-space models with unknown parameters and non-Gaussian measurements.
}


\subsection{Hierarchical state-space models \label{sec:hssmintro}}

\textcolor{black}{
Here, we develop new methods to analyze a general class of high-dimensional \emph{hierarchical state-space models} \citep[HSSMs; e.g.,][]{Cressie2011}.   These models \textcolor{black}{consist of four levels, adding} two levels to a standard Gaussian SSM: a \emph{transformation level}, which allows for non-Gaussian observations, and a \emph{parameter level}, which models the unknown parameters, \textcolor{black}{assumed to be random}. This class of HSSMs is very general and can be used to describe most dynamic spatio-temporal systems observed sequentially over time. It allows for unknown time-varying parameters in any part of the model, non-Gaussian and nonlinear observations, mixture models, and even higher-order Markov models by extending the state-space.}

More specifically, for discrete time points $t=1,2,\ldots,T$, our HSSM is given by:
\begin{align}
\textcolor{black}{\mbox{Transformation:}} ~~~~~~ \bz_t | \by_t, \bftheta_t & \sim f_t(\bz_t|\by_t,\bftheta_t), \label{eq:transmodel}\\
\textcolor{black}{\mbox{Observation:}} ~~~~~~ \by_t | \bx_t, \bftheta_t & \sim \normal_{m_t}(\bH_t\bx_t,\bR_t),  \label{eq:obsmodel}\\
\textcolor{black}{\mbox{Evolution:}} ~~~ \bx_t | \bx_{t-1}, \bftheta_t & \sim \normal_n(\forex_t \! = \! \evol_t(\bx_{t-1}),\bQ_t), \label{eq:evomodel}\\
\textcolor{black}{\mbox{Parameter:}} ~~~~~~~ \bftheta_t | \bftheta_{t-1} & \sim \textcolor{black}{p}(\bftheta_t| \bftheta_{t-1}), \label{eq:paramodel}
\end{align}
\textcolor{black}{
where $\bz_t$ is the $m_t$-dimensional measurement vector at time $t$, $\by_t$ is a corresponding $m_t$-dimensional ``latent observation'' vector, $\bx_t$ is the $n$-dimensional unobserved state vector, $\bftheta_t$ is the parameter vector, $f_t$ is a known measurement or transformation distribution, $\evol_t$ is the state evolution operator, and $\normal_n(\cdot,\cdot)$ denotes the $n$-variate normal distribution.  We assume that $\bz_t$ and $\by_t$ are serially independent given the states and parameters.  The model is completed with initial priors $\bx_0 \sim \normal_n(\bfmu_{0|0},\bfSigma_{0|0})$ and $\bftheta_0 \sim \textcolor{black}{p}(\bftheta_0)$.}

\textcolor{black}{
The state vector $\bx_t$ consists of one or more discretized spatial fields, along with any parameters that can be estimated using state augmentation.  Conversely, the parameter vector $\bftheta_t$ contains any unknown quantities that {\em cannot} be handled well by state augmentation.  The quantities $\bH_t$, $\evol_t$, $\bR_t$, $\bQ_t$, and $f_t$ may all depend on the parameter vector $\bftheta_t$. Sometimes we make the dependence explicit by writing $\bH_t(\bftheta_t)$ and so forth.  Further, we view the evolution operator $\evol_t$ as a ``black box" that might be nonlinear and expensive to evaluate or unavailable in closed form. We focus on models where $n$ or $m_t$ are very large, \textcolor{black}{usually $n > m_t$.
Sometimes, we observe $\bz_t=\by_t$ so that $f(\bz_t|\by_t,\bftheta_t)=\delta_{\by_t}(\bz_t)$ where $\delta(\cdot)$ is the Dirac delta function.}}

\textcolor{black}{The HSSM \eqref{eq:transmodel}--\eqref{eq:paramodel} is \textcolor{black}{very} flexible and includes many important SSMs as special cases.  If $\by_t$ and $\bftheta_t$ are known, the model reduces to a standard Gaussian SSM given by (2)--(3) \textcolor{black}{\cite[e.g.,][]{Shum:Stof:06}}.  If the evolution is linear with $\evol_t(\bx_{t-1})=\bM_t\bx_{t-1}$, \textcolor{black}{this} becomes a {\em conditionally Gaussian SSM} \textcolor{black}{\cite[e.g.,][]{Carter1994,Chen2000}}, where the full conditional distributions for the states are Gaussian. The HSSM also allows for non-Gaussian and nonlinear observations through $f_t$ in \eqref{eq:transmodel}. This could include exponential family \textcolor{black}{models \cite[e.g.,][]{Shephard1997,Gamerman1998}}, nonlinear \textcolor{black}{mean functions \cite[e.g.][]{Carlin1992}}, or threshold models for \textcolor{black}{discrete data or discrete-continuous mixtures \cite[e.g.,][]{SansoGuenni1999}}.  Finally, the HSSM allows for \textcolor{black}{unknown static or time-varying parameters} at any level of the model.   
}

\textcolor{black}{Note that \textcolor{black}{it may be} possible to reduce our HSSM to a standard nonlinear, non-Gaussian SSM by integrating out $\by_t$ and defining $(\bx_t,\bftheta_t)$ as the state vector. \textcolor{black}{However, we} prefer the hierarchical formulation in \eqref{eq:transmodel}--\eqref{eq:paramodel}, because it facilitates model building and computation. Specifically, by conditioning on $\by_t$ and $\bftheta_t$, we will use the EnKF for inference on the high-dimensional states.}


{\color{black}
\subsubsection{Examples \label{sec:allexamples}}

We now present some examples of the HSSM in \eqref{eq:transmodel}--\eqref{eq:paramodel} that will be used for numerical comparisons in Section \ref{sec:examples}.

\begin{example}[Heavy-tailed-data model]\label{ex:tdist}
The distribution of $\by_t|\bx_t$ belongs to the class of normal-scale mixtures \citep{Andrews1974,West1987} if we assume that $\bR_t$ is diagonal with independent random parameters $\bftheta_t$ on the diagonal. An important member of this class is the $t$-distribution, which is similar to a normal distribution but has heavier tails and is thus more robust to outliers. 
The following HSSM assumes that the observation noise is independently $t$-distributed with $\kappa$ degrees of freedom:
\begin{align*}
\bz_t | \by_t & \sim \delta_{\by_t}(\bz_t),\\
\by_t | \bx_t, \bftheta_t & \sim \normal_{m_t}(\bH_t\bx_t,\bR_t(\bftheta_t)),\\
\bx_t | \bx_{t-1} & \sim \normal_n(\evol_t(\bx_{t-1}), \bQ_t),\\
\theta_{t,l} & \stackrel{iid}{\sim} \ig(\kappa/2,\kappa/2), \quad l=1,\ldots,m_t,
\end{align*}
where $\bR_t(\bftheta_t) = \sigma^2_t \diag(\theta_{t,1},\ldots,\theta_{t,m_t})$, $\sigma^2_t$ is known, and $\bftheta_t=(\theta_{t,1},\ldots,\theta_{t,m_t})'$ follow independent inverse-gamma distributions.  The transformation level is a point mass such that $\bz_t=\by_t$. 
\end{example}


Next, we consider non-continuous observations, such as binary data and right-skewed data with a point mass at zero (e.g., hourly rainfall amounts).
The EnKF and its transformation- or quantile-based extensions \citep[e.g.,][]{Anderson2009,Amezcua2014} are particularly poorly suited for assimilation of  such observations. In our HSSM framework, this setting is handled using threshold observation models, in which $f_t$ is determined by transformations of the form $z_{t,l} = g_t(y_{t,l}; \bftheta_t)$, $l=1,\ldots,m$, where $g_t(\cdot)$ is a deterministic function containing an indicator function:
\begin{example}[Threshold observation model for rainfall] \label{ex:rainfall} 
The observed data are modeled as $z_{t,l} = g(y_{t,l}; \bftheta_t) = y_{t,l}^{\theta_t} \mathbbm{1}_{\{y_{t,l} > 0\}},$ where $\mathbbm{1}$ is the indicator function.  Then $z_{t,l}$ follows a rainfall-type distribution, i.e., a mixture of a positive, right-skewed distribution and a point mass at zero, for some $\theta_t>1$ that determines the skewness \citep[e.g.,][]{SansoGuenni1999,Sigrist2011}. The HSSM is given as:
\begin{align*}
z_{t,l} | \by_t, \theta_t & = y_{t,l}^{\theta_t} \mathbbm{1}_{\{y_{t,l}>0\}}, \quad l=1,\ldots,m_t,\\
\by_t | \bx_t & \sim \normal_{m_t}(\bH_t\bx_t,\sigma^2_t\bI_{m_t}),\\
\bx_t | \bx_{t-1} & \sim \normal_n(\evol_t(\bx_{t-1}), \bQ_t),\\
\theta_t & \sim p(\theta_t),
\end{align*}
where the parameter $\theta_t$ is either known or assigned a prior distribution.
\end{example}
The threshold framework can also be used for \emph{binary} data following a Bernoulli distribution with $z_{t,l} | \bx_t, \bftheta_t \sim \mathcal{B}ern\left(\Phi((\bH_t \bx_t)_l/\sigma_t)\right)$, where $\Phi(\cdot)$ is the standard normal CDF. This is called a probit model for $\sigma_t=1$.  The transformation function is $z_{t,l} = g(y_{t,l}; \bftheta_t) = \mathbbm{1}_{\{y_{t,l} > 0\}}$, which corresponds to the rainfall model with parameter $\theta_t=0$.

\begin{example}[Dynamic Poisson model]\label{ex:poisson}
For spatio-temporal count data, assume $f_t(\bz_t|\by_t,\bftheta_t) = \prod_{l=1}^{m_t} \mathcal{P}ois(z_{t,l}|e^{y_{t,l}})$ and $\bR_t = \sigma^2_{t} \bI_{m_t}$, where $\sigma^2_t$ accounts for overdispersion \citep{Wikl:02}.
The state follows a linear evolution, and the model includes six unknown time-varying parameters:
\begin{align*}
z_{t,l} | \by_t & \sim \mathcal{P}ois(\exp(y_{t,l})), \quad l=1,\ldots,m_t,\\
\by_t | \bx_t, \bftheta_t & \sim \normal_{m_t}(\bH_t\bx_t,\sigma^2_t\bI_{m_t}),\\
\bx_t | \bx_{t-1},\bftheta_t & \sim \normal_n(\bM(\bfgamma_t)\bx_{t-1},\bQ(\tau^2_t,\lambda_t)),\\
\bftheta_t |\bftheta_{t-1} & \sim \normal_6(\bftheta_{t-1},.05^2\bI_6),
\end{align*}
where the transition matrix $\bM(\bfgamma_t)$ is tridiagonal with parameters $\bfgamma_t=(\gamma_{1,t},\gamma_{2,t},\gamma_{3,t})$ and $\bQ(\tau^2_t,\lambda_t)$ is a spatial covariance matrix with variance $\tau^2_t$ and range $\lambda_t$. The unknown time-varying parameters $\bftheta_t=(\gamma_{1,t},\gamma_{2,t},\gamma_{3,t},\allowbreak\log \sigma_t,\allowbreak \log \tau_t, \log \lambda_t)'$ follow a random-walk evolution. 
\end{example}

\begin{example}[Nonlinear evolution model with unknown static parameter] \label{ex:lorenz}
The Lorenz-96 model is widely used in atmospheric science as an example of a nonlinear evolution operator 
with chaotic behavior \citep{Lore:96}.  Lorenz-96 mimics advection at $n=40$ equally-spaced locations along a latitude circle on the globe. 
The time evolution is determined by 
\begin{equation}
\label{eq:lorenz}
\textstyle\frac{d x_{t,i}}{dt} = (x_{t,i+1}-x_{t,i-2})x_{t,i-1}-x_{t,i} + F, \qquad i=1,\ldots,40,
\end{equation}
with cyclic boundary conditions (i.e., $x_{t,38}=x_{t,-2}$, $x_{t,39}=x_{t,-1}$, $x_{t,41}=x_{t,1}$). Let $\lor_{8, 0.2}$ denote the evolution operator implied by \eqref{eq:lorenz} with forcing parameter $F=8$ and time step \textcolor{black}{$\Delta t=0.2$,} which leads to strong nonlinearities and non-Gaussian forecast distributions.  We observe data $\bz_t=\by_t$ at each location so that $m=n=40$.  The HSSM model is:
\begin{align*}
\bz_t | \by_t & \sim \delta_{\by_t}(\bz_t),\\
\by_t | \bx_t & \sim \normal_n(\bx_t, \sigma^2\bI_n),\\
\bx_t | \bx_{t-1},\theta & \sim \normal_n(\evol(\bx_{t-1};\theta),\bQ),\\
\theta & \sim \normal(\mu_\theta,\sigma^2_\theta),
\end{align*}
where $\sigma^2$ is known, $\evol_t(\bx_{t-1};\theta) = \theta \, \lor_{8,0.2}(\bx_{t-1})$ is the nonlinear evolution operator, which is unavailable in closed form, $\theta$ is an unknown parameter, and $\bQ$ is a known covariance matrix.
\end{example}
}

\subsection{Extended ensemble Kalman \textcolor{black}{techniques}}

\textcolor{black}{
\textcolor{black}{The goal in this paper is to make Bayesian inference on the states $\bx_t$ and parameters $\bftheta_t$ for HSSMs \eqref{eq:transmodel}--\eqref{eq:paramodel}.} \textcolor{black}{Let $\bz_{1:t}=\{\bz_1,\ldots,\bz_t\}$ denote the data up to time $t$, and use similar notation for $\bx, \by$ and $\bftheta$.}  Let \textcolor{black}{$p(x|z)$} denote the conditional distribution of some random variable $x$ given $z$.   We consider two types of inference.  {\em Filtering} involves computing the filtering distribution \textcolor{black}{$p(\bx_t,\bftheta_t|\bz_{1:t})$} sequentially for each time $t=1,2,\ldots$. {\em Smoothing} requires calculating the joint posterior (i.e., smoothing) distribution \textcolor{black}{$p(\bx_{1:T},\bftheta_{1:T}|\bz_{1:T})$} for fixed time $T$.}




We propose a \emph{simple general idea} for such inference in \textcolor{black}{high-dimensional HSSMs}: Take \textcolor{black}{an} existing Bayesian inference technique (e.g., a particle filter or a Gibbs sampler), but approximate the part of the algorithm that requires sampling from or integrating out the high-dimensional state vector $\bx_t$ with an EnKF or with the related ensemble Kalman smoother (EnKS). This leads to a general class of \emph{extended ensemble Kalman filters and smoothers}, which can be applied to \textcolor{black}{high-dimensional} HSSMs. The key is that when conditioning on $\by_t$ and $\bftheta_t$, our HSSM reduces to a standard Gaussian SSM with potentially nonlinear evolution, for which the EnKF is often well suited.  Table~\ref{tab:methods} summarizes these algorithms and puts them in the context of existing particle approaches or approaches for linear Gaussian SSMs.
 


{\color{black}
Our proposed filtering and smoothing algorithms can be classified \textcolor{black}{into {\em Ensemble Gibbs} and {\em Ensemble Marginalization} methods}.   For \textcolor{black}{Ensemble} Gibbs schemes, the EnKF update is used within a Gibbs sampler to draw from the full conditional distribution of $\bx$ given $\bftheta$ and $\by$. For \textcolor{black}{Ensemble Marginalization} schemes, the EnKF is used to integrate out {\em and} sample from the posterior distribution of the states $\bx$ within a particle filter for the parameters $\bftheta$, similar to a mixture KF \citep{Chen2000}. Further connections to existing methods will be provided with each of our proposed algorithms.}

Our new extended \textcolor{black}{EnK techniques} can be viewed from two perspectives. \textcolor{black}{First, these methods greatly expand the class of problems to which the EnKF and EnKS can be applied.  Second, our techniques allow for accurate, approximate inference in higher-dimensional models than standard statistical methods such as particle filters and Gibbs samplers.}


\textcolor{black}{In the EnKF literature, several approaches exist for estimating specific tuning parameters, such as inflation and localization parameters 
\citep[e.g.,][]{WangBish:03,Ande:07a}. Throughout, we assume that these tuning parameters are known, and we focus instead on inference of model parameters that explicitly appear in the model \eqref{eq:transmodel}--\eqref{eq:paramodel}.} \textcolor{black}{Some approaches for state and parameter estimation in highly nonlinear state-space models \citep{Gu2007,Chen2012,Bocquet2013} generate approximate posterior samples from the filtering or smoothing distribution using iterative optimization approaches. Specifically, the ensemble randomized maximum likelihood method \citep{Chen2012} recasts posterior simulation as an optimization problem, and uses optimization methods with successive linearization of the model to find the posterior mode. These methods are quite general and can be used for parameter estimation and smoothing inference via state augmentation. 
}

As the EnKF is an approximate technique for inference on the state vector, our extended \textcolor{black}{EnK techniques} also \textcolor{black}{provide approximate inference}. However, in many complex problems, exact inference is impossible, and approximate techniques such as approximate Bayesian computation \citep[ABC;][]{Sisson2007,Marin2012}, synthetic likelihood \citep[e.g.,][]{Wood2010} and integrated nested Laplace approximations \citep[INLA;][]{Rue2009,Lindgren2011a} have generally gained in popularity.  While so-called exact approximate approaches \citep[e.g.,][]{Andrieu2010,Johansen2012} result in exact inference given infinite computational resources, we will show here numerically that, in finite computation time, our extended \textcolor{black}{EnK techniques} can outperform exact approximate \textcolor{black}{methods} even for moderately high dimensions.

\begin{table}
\centering
\small
\begin{tabular}{p{.17\textwidth}| p{.11\textwidth} | p{.2\textwidth} | p{.18\textwidth} | p{.22\textwidth}}
 & & \multicolumn{3}{c}{\large Filtering/smoothing methods}\\
Target distrib. & Dim.\ of $\bftheta_t$ & \textcolor{black}{Cond.\ }Gaussian    & Particle & Ensemble \\
\hline \hline
\textcolor{black}{$p(\bx_t|\bftheta_{1:t},\by_{1:t})$}, $\forall t$  & ($\bftheta$ known) & Kalman filter  &  Particle filter   & EnKF (Alg.~\ref{alg:enkf})\\
\hline
\textcolor{black}{$p(\bx_{1:T}|\bftheta_{1:T},\by_{1:T})$}        & ($\bftheta$ known)  & \textcolor{black}{FFBS}  & Particle smoother & EnKS (Alg.~\ref{alg:enks})  \\
\hline
\multirow{2}{*}{
\textcolor{black}{$p(\bx_t,\bftheta_t|\bz_{1:t})$}, $\forall t$} & high &  RWFFBS & Online MCMC & \textbf{GEnKF (Alg.~\ref{alg:genkf})}  \\
& low & RBPF & EARBPF &  \textbf{PEnKF (Alg.~\ref{alg:penkf})} \\
 \hline
\multirow{3}{*}{\textcolor{black}{$p(\bx_{1:T},\bftheta_{1:T}|\bz_{1:T})$}}  
& high & Gibbs with FFBS  & Particle Gibbs  & \textbf{GEnKS (Alg.~\ref{alg:genks})} \\
& low & MH with KF  & PMMH    & \textbf{MHEnKS}/ \\
&       &                      &   &  \textbf{PEnKS (Alg.~\ref{alg:penks})}
 \end{tabular}
\caption{\label{tab:methods} 
Existing and our new (in bold) methods for filtering and smoothing tasks for known and unknown parameters. The algorithms for high-dimensional $\bftheta_t$ generally require its full-conditional distribution to be available in closed form. ``\textcolor{black}{Cond.\ Gaussian'' refers to methods where the full conditional distribution for the states is Gaussian}. The methods in the column ``Particle'' do not scale to large state dimension $n$. Abbreviations and main references: KF = Kalman filter \citep{Kalman1960}; particle filter \citep{Gordon1993}; EnKF = ensemble Kalman filter \citep{Even:94}; FFBS = forward-filtering backward-sampling \citep{Carter1994}; particle smoother \citep{Kita:96}; EnKS = ensemble Kalman smoother \citep{Even:vanL:00}; RWFFBS = rolling-window FFBS \citep{polson2008practical}; Online MCMC \citep{berzuini1997dynamic}; RBPF = Rao-Blackwellized PF \citep{Doucet2000}; EARBPF = exact approximate RBPF \citep{Johansen2012}; particle Gibbs \citep{Andrieu2010}; MH with KF = Metropolis-Hastings based on KF likelihood \citep{schweppe1965evaluation}; PMMH = particle marginal Metropolis-Hastings \citep{Andrieu2010}. } 
\end{table}

\subsection{Overview and contributions}

This article is organized as follows. In Section \ref{sec:enkf}, we review some existing work on the EnKF and related issues from a statistical perspective. In Section \ref{sec:preliminaries}, we provide preliminaries to lay the foundation for the development of our extended \textcolor{black}{EnK techniques}, including new insights on likelihood approximations in high-dimensional SSMs. In Section \ref{sec:methods}, we derive and present our extended ensemble Kalman filters and smoothers. In Section \ref{sec:examples}, \textcolor{black}{we compare our methods to existing approaches using a number of simulated and real data examples.  We find that our methods strongly outperform these existing approaches.}  Conclusions are given in Section \ref{sec:conclusions}.


{\color{black}
\section{Review of ensemble Kalman filter and smoother\label{sec:enkf}}

In this section, we briefly review the ensemble Kalman filter (EnKF) and its extensions.  The EnKF is a sequential Monte Carlo algorithm that generates samples from the state filtering distribution for the SSM (2)-(3), where $\by_{1:t}$ and $\bftheta_{1:t}$ are known. Thus, throughout this section, we assume $\by_{1:t}$ and $\bftheta_{1:t}$ are known and we often suppress $\bftheta_{1:t}$ from the conditioning set.  Further, let $\normal_n(\bx|\bfmu,\bfSigma)$ denote an $n$-variate normal density with mean $\bfmu$ and covariance matrix $\bfSigma$, evaluated at $\bx$.
}


As mentioned in Section \ref{sec:hssmintro}, we assume here that the evolution operator $\evol_t$ is intractable or computationally expensive, and so all we are able to do with $\evol_t$ is to ``apply'' it to a finite and typically rather small number of state vectors. Thus, sequential inference over time must be based on an \emph{ensemble} of state vectors that represents (i.e., approximates) the relevant distributions at each time point.

More specifically, assume that the \emph{filtering ensemble} at time $t-1$ is given by $\bx_{t-1|t-1}^{(1:N)} \colonequals \{\bx_{t-1|t-1}^{(1)},\ldots,\bx_{t-1|t-1}^{(N)}\}$, meaning that the filtering distribution at time $t-1$ is approximated as 
\begin{equation}
\label{eq:priornopar}
\textcolor{black}{p(\bx_{t-1} | \by_{1:t-1})} = \textstyle\frac{1}{N} \textstyle\sum_{j=1}^N \delta_{(\bx_{t-1|t-1}^{(j)})}(\bx_{t-1}),
\end{equation}
where $\delta$ is the Dirac delta (point mass) function. The \emph{forecast ensemble} {\color{black}$\forex_{t|t-1}^{(1:N)} = \{\forex_{t|t-1}^{(1)},\ldots,\forex_{t|t-1}^{(N)}\}$} is obtained by applying {\color{black}$\evol_t$} to each member of the filtering ensemble:
\begin{equation}
\label{eq:enkfforecast}
\forex_{t|t-1}^{(j)} = \evol_t(\bx_{t-1|t-1}^{(j)}), \qquad j=1,\ldots,N.
\end{equation}

\subsection{Updating the state \label{sec:enkfreview}}

Based on new data $\by_t$ observed at time $t$, we would like to 
{\color{black} update the forecast ensemble in \eqref{eq:enkfforecast} to obtain a filtering ensemble from the filtering distribution of the state vector, \textcolor{black}{$p(\bx_t|\by_{1:t})$}.}

\subsubsection{{\color{black}Update with a mixture prior}}

Based on \eqref{eq:priornopar}, the prior distribution at time $t$ would in principle be given by
\begin{equation}
\label{eq:forecast1}
\textcolor{black}{\textstyle p(\bx_t | \by_{1:t-1}) = \int p(\bx_t|\bx_{t-1}) dP(\bx_{t-1}|\by_{1:t-1}) = \textstyle\frac{1}{N} \sum_{j=1}^N \normal_n(\bx_t|\forex_{t|t-1}^{(j)},\bQ_t),}
\end{equation}
where the integral is with respect to the measure induced by the distribution of $\bx_{t-1}|\by_{1:t-1}$. However, based on the mixture prior distribution in \eqref{eq:forecast1}, the filtering distribution updated using the new observation $\by_t$ would have the form
\begin{align}
\textcolor{black}{p(\bx_t | \by_{1:t})} & \propto \textcolor{black}{p(\by_t|\bx_t)p(\bx_t|\by_{1:t-1})} = \textstyle\frac{1}{N} \sum_{j=1}^N \normal_{m_t}(\by_t|\bH_t\bx_t,\bR_t)\textcolor{black}{\normal_n(\bx_t|\forex_{t|t-1}^{(j)},\bQ_t)} \nonumber\\
& \propto \textstyle\sum_{j=1}^N \alpha_t^{(j)} \textcolor{black}{p(\bx_t|\by_t,\bx^{(j)}_{t-1|t-1})} \label{eq:naivefilter},
\end{align}
with the filtering or posterior weights 
\begin{equation}
\label{eq:pfweights}
\alpha_t^{(j)} = \textcolor{black}{p(\bx^{(j)}_{t-1|t-1}| \by_{1:t})} \propto \textcolor{black}{\normal_{m_t}(\by_t|\bH_t\forex_{t|t-1}^{(j)},\bH_t\bQ_t\bH_t'+\bR_t).}
\end{equation}
This reweighting of the $j$th forecast ensemble member \textcolor{black}{$\forex_{t|t-1}^{(j)}$} from prior weight $1/N$ to posterior weight $\alpha_t^{(j)}$ is essentially what happens in an importance sampler or in the update step of a particle filter that integrates out the innovation error with covariance $\bQ_t$. This reweighting update is illustrated in a toy example in Figure \ref{fig:updatingillustration} for the simple case $m_t=1$, $\bH_t=1$, and $\bQ_t = \bfzero$. However, as shown in \citet[][]{Snyd:Beng:Bick:08} by setting their observation error covariance matrix $\bR$ to be $\bH_t\bQ_t\bH_t'+\bR_t$ in our notation, the particle weights \emph{degenerate}, meaning that one mixture component would get all or nearly all of the weight, unless the ensemble size $N$ increases exponentially with the effective dimension given by the variance of the particle loglikelihood (see Section \ref{sec:pflik} below). Thus, in our high-dimensional setting where only a small ensemble size $N$ is computationally affordable, this update based on reweighting is problematic.

\begin{figure}
\begin{minipage}[t]{.09\textwidth}
~\\
%
%
%
\vspace{-44mm}

posterior\\
\vspace{5mm}

likelihood\\

\vspace{5mm}
prior
\end{minipage}
\hfill
\begin{minipage}[t]{.28\textwidth}
	\centering 
	\includegraphics[width=.98\linewidth]{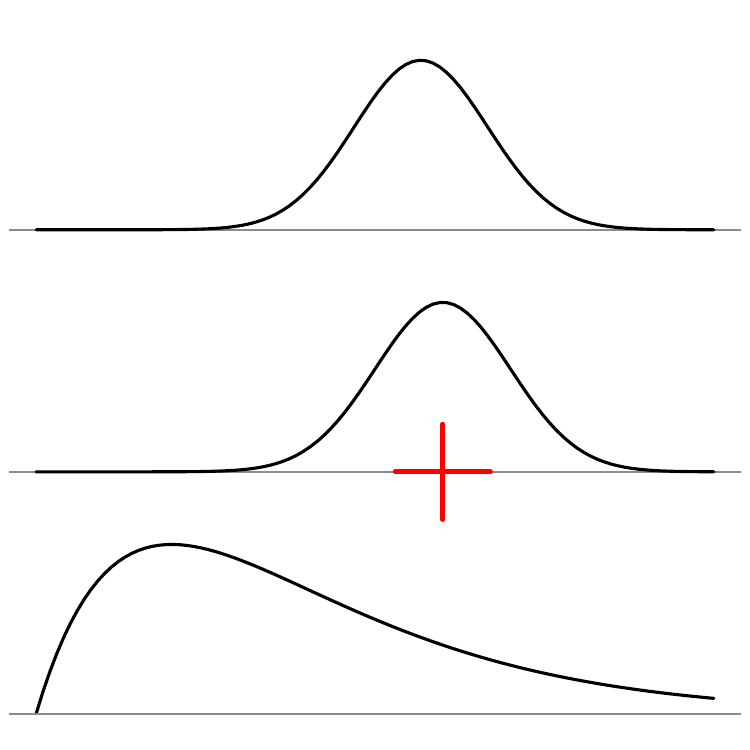}\\
\small Exact
\end{minipage}
\hfill
\begin{minipage}[t]{.28\textwidth}
	\centering 
	\includegraphics[width=.98\linewidth]{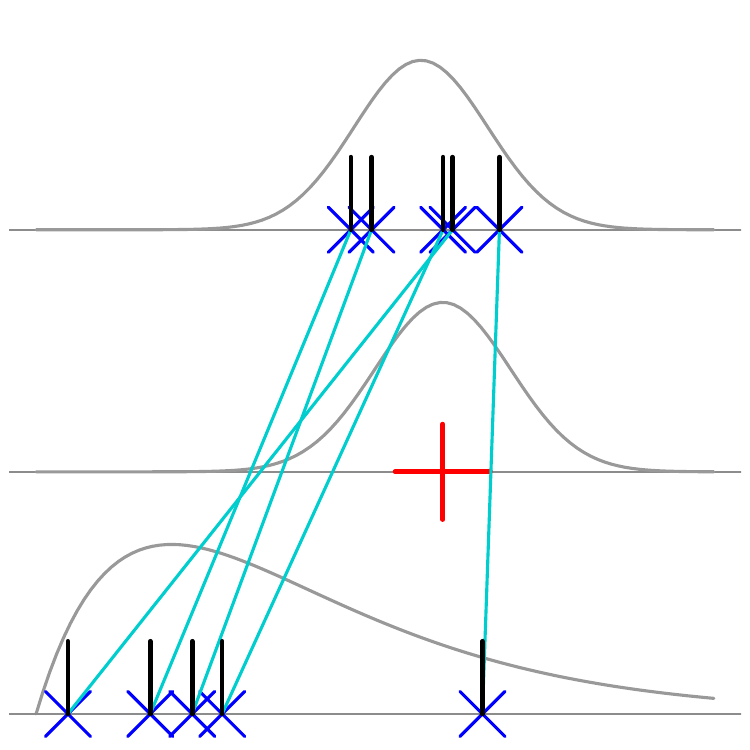}\\
\small EnKF 
\end{minipage}
\hfill
\begin{minipage}[t]{.28\textwidth}
	\centering 
	\includegraphics[width=.98\linewidth]{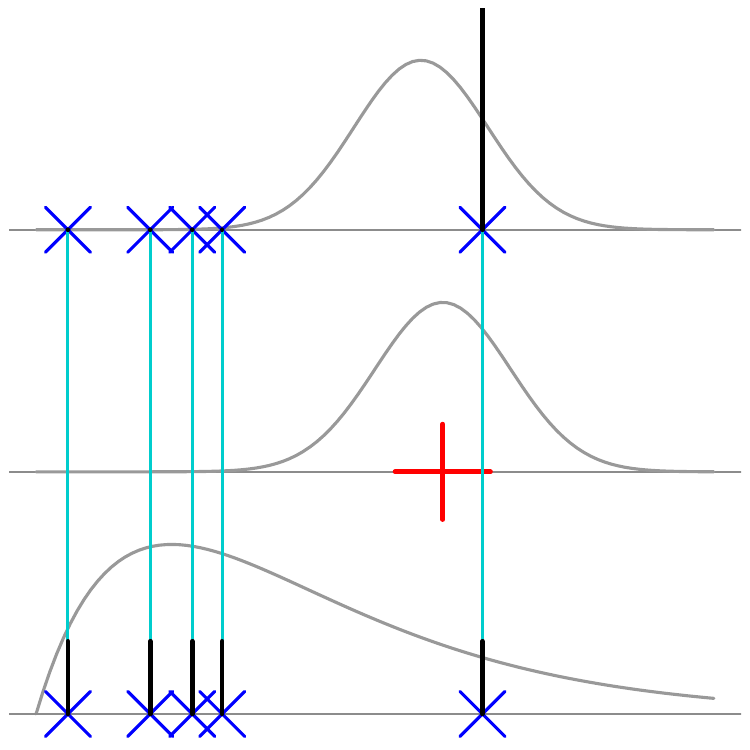}\\
\small \color{black}{Particle Filter} 
\end{minipage}
\caption{Comparison of \textcolor{black}{updating methods} in a simple one-dimensional example with a non-Gaussian prior distribution. The bars above the ensemble/particles are proportional to the weights. Both approximation methods start with the same, equally\textcolor{black}{-}weighted $N=5$ prior samples. Even for this one-dimensional example, the \textcolor{black}{importance weights degenerate for the particle filter}, while the EnKF shifting update obtains a better representation of the posterior.}
\label{fig:updatingillustration}
\end{figure}

\subsubsection{{\color{black}Update with a Gaussian prior}}

In contrast \textcolor{black}{to the particle filter}, the EnKF update does not reweight but rather shifts the forecast ensemble. The form of the shift is motivated by the fact that if the evolution \textcolor{black}{operator $\evol_t$ is linear at each time $t$}, then the true prior distribution will be \textcolor{black}{Gaussian}. In that case, it is natural to approximate the mixture of Gaussians in the prior distribution \eqref{eq:forecast1} by a \textcolor{black}{single} Gaussian distribution:
\begin{equation}
\label{eq:forecast2}
\textcolor{black}{\textstyle p(\bx_t | \by_{1:t-1})= \normal_n(\bx_t|\widehat\bfmu_{t|t-1},\widehat\bfSigma_{t|t-1}),}
\end{equation}
where $\widehat\bfmu_{t|t-1}$ and $\widehat\bfSigma_{t|t-1}$ are estimates of the prior mean and covariance matrix (see Section \ref{sec:tapering} below) based on the forecast ensemble obtained in \eqref{eq:enkfforecast}. Given this Gaussian prior, the posterior would also be Gaussian:
\begin{equation}
\label{eq:enkfpost}
\textcolor{black}{p(\bx_t | \by_{1:t}) \propto p(\by_t|\bx_t) p(\bx_t|\by_{1:t-1})} = \normal_{m_t}(\by_t|\bH_t\bx_t,\bR_t) \normal_n(\bx_t|\widehat{\bfmu}_{t|t-1},\widehat{\bfSigma}_{t|t-1}) \propto \normal_n(\bx_t|\widehat{\bfmu}_{t|t},\widehat{\bfSigma}_{t|t}),
\end{equation}
where $\widehat{\bfmu}_{t|t}=\widehat{\bfmu}_{t|t-1}+\widehat{\bK}_t(\by_t-\bH_t\widehat{\bfmu}_{t|t-1})$, $\widehat{\bfSigma}_{t|t}=(\bI - \widehat\bK_t\bH_t)\widehat{\bfSigma}_{t|t-1}$, and 
\begin{equation}
\label{eq:gain}
\widehat\bK_t = \widehat{\bfSigma}_{t|t-1} \bH_t'(\bH_t\widehat{\bfSigma}_{t|t-1}\bH_t'+\bR_t)^{-1}
\end{equation}
is an estimate of the Kalman gain, $\bK_t$.

\subsubsection{The \textcolor{black}{ensemble Kalman filter} update}

{\color{black}
In the stochastic EnKF update, instead of calculating the posterior distribution in \eqref{eq:enkfpost} explicitly, a sample from this distribution is obtained via \emph{shifting} the forecast ensemble members by essentially carrying out conditional simulation \citep[e.g.,][]{Katzfuss2015b}. This leads to the following algorithm:
\begin{framed}
\begin{algorithm}
\label{alg:enkf} \textbf{Ensemble Kalman filter (EnKF)}\\
Start with an initial ensemble $\bx_{0|0}^{(1)},\ldots,\bx_{0|0}^{(N)} \sim \normal(\bfmu_{0|0},\bfSigma_{0|0})$. 
Then, at each time $t=1,\ldots,T$, do the following for each $j =1,\ldots,N$:
\begin{description}
\item[1. Forecast Step:] Compute forecast and prior ensemble members, $\forex_{t|t-1}^{(j)} = \evol_t(\bx_{t-1|t-1}^{(j)})$ and $\bx_{t|t-1}^{(j)} = \forex_{t|t-1}^{(j)} + \bw_t^{(j)}$, respectively, where $\bw_t^{(j)} \sim \normal_n(\bfzero,\bQ_t)$.
\item[2. Update Step:]  Shift each ensemble member as
\begin{equation}
\bx_{t|t}^{(j)}  = \bx_{t|t-1}^{(j)} + \widehat{\bK}_t(\by_t - \widetilde\by_t^{(j)}) \equalscolon \update_t(\forex_{t|t-1}^{(j)} | \by_t,\forex_{t|t-1}^{(1:N)},\bftheta_t),
\label{eq:enkfupdate}
\end{equation}
where $\widehat{\bK}_t$ is defined in \eqref{eq:gain}, and $\widetilde\by_t^{(j)}  = \lobs_t\bx_{t|t-1}^{(j)} + \bv_{t}^{(j)}$ is a pseudo-observation with $\bv_t^{(j)} \stackrel{iid}{\sim} \normal(\bfzero,\bR_t)$.
\end{description}
\end{algorithm}
\end{framed}
The EnKF shifting update $\update_t$ in \eqref{eq:enkfupdate} is derived based on assuming $\evol_t$ is linear and hence that the prior distribution is multivariate Gaussian. The forecast ensemble $\forex_{t|t-1}^{(1:N)}$ enters $\update_t$ through the estimated Kalman gain $\widehat{\bK}_t$, and $\bH_t$, $\bQ_t$, and $\bR_t$ are evaluated at the parameter value $\bftheta_t$, which is currently still assumed to be known. 
}
The resulting filtering distribution converges to the true filtering distribution as $N\rightarrow\infty$, assuming the estimator $\bfSigma_{t|t-1}$ converges to the true forecast covariance (see Section \ref{sec:tapering} below). For finite $N$, in contrast to many other SMC methods, the EnKF update does not deteriorate for large $n$, as long as a reasonable estimate of $\bK_t$ or $\bfSigma_{t|t-1}$ can be found. It has been shown repeatedly \citep[e.g.,][]{Lei2010} that this shifting update also works surprisingly well for nonlinear evolution or non-Gaussian priors (see Figure \ref{fig:updatingillustration} for an illustration), including operational, real-world applications \citep[e.g.,][]{Houtekamer2005,Bonavita2010,Houtekamer2014}.

\subsection{Estimation of matrices \label{sec:tapering}}

The EnKF requires estimates of the mean {\color{black}vector} $\widehat\bfmu_{t|t-1}$ and covariance matrix $\widehat\bfSigma_{t|t-1}$ of the prior distribution in \eqref{eq:forecast2}. The prior covariance {\color{black} matrix} is required for the EnKF update through the Kalman gain \eqref{eq:gain}, and both the mean and covariance {\color{black} matrix} are needed in the EnKF likelihood in Section \ref{sec:elik} below.  The prior mean can simply be estimated by the sample mean of the forecast ensemble; that is, $\widehat\bfmu_{t|t-1} = \frac{1}{N} \sum_{j=1}^N \forex_{t|t-1}^{(j)}$.

Estimation of the $n \times n$ prior covariance matrix based on small ensembles is often challenging in applications with large state dimension $n$,
{\color{black}
which can lead to ensemble collapse \textcolor{black}{\citep[e.g.,][]{Saetrom2013}} similar to the degeneration of the particle filter. In the EnKF, this collapse is counteracted by variance inflation \citep[e.g.,][]{Ande:07a}, data-dimension-reduction via shrinkage regression \textcolor{black}{\citep{Saetrom2011}}, or regularization via tapering. We focus here on \emph{tapering} \citep[e.g.,][]{Furr:Gent:Nych:06,Ande:07,Furr:Beng:07}, which is a very powerful method to avoid rank deficiency and spurious correlations by assuming that dependence is expected to decrease with increasing spatial distance. More precisely, tapering}
refers to entrywise multiplication with a sparse positive definite correlation matrix defined based on the distances between the elements. This leads to the expression 
{\color{black}
$\widehat\bfSigma_{t|t-1} = \widetilde\bfSigma_{t|t-1} + \bQ_t$ with $\widetilde\bfSigma_{t|t-1} = \bC_{t} \circ \mathbf{\mathcal{T}}_t$, where $\bC_{t}$
}
is the sample covariance matrix of $\{\forex_{t|t-1}^{(1)},\ldots,\forex_{t|t-1}^{(N)}\}$, and $\mathbf{\mathcal{T}}_t$ is the sparse correlation matrix. Other regularization methods for states that are spatially referenced \citep[e.g.,][]{Ott:Hunt:Szun:04,Hunt2007} or not \citep[e.g.,][]{Pourahmadi2013} will not be discussed here, but could also be applied in some of our methods described later. 
{\color{black}
Also note that $\widehat\bfSigma_{t|t-1}$ enters the Kalman gain in \eqref{eq:gain} through the $n \times m_t$ matrix $\widehat{\bfSigma}_{t|t-1} \bH_t'$ and the $m_t \times m_t$ matrix $\bH_t\widehat{\bfSigma}_{t|t-1}\bH_t'$, and so the EnKF can be stable even for massive $n$, as long as $m_t$ is not too large.
} 

\subsection{Ensemble Kalman smoother \label{sec:EnKSreview}}

To approximate the smoothing distribution \textcolor{black}{$p(\bx_t|\by_{1:T})$} for some $T$, various extensions of the EnKF to an ensemble Kalman smoother have been proposed \citep{vanL:Even:96,Even:vanL:00,Khar:Ande:08,Stro:Stei:Lesh:10,Bocq:Sako:14}.
The most widely used smoother is the EnKS algorithm of \cite{Even:vanL:00}. Like the EnKF, the EnKS is a forward-only algorithm that is applicable even for intractable evolution operators $\evol_t$. At each time $t$, it essentially carries out an EnKF update on the augmented state that includes the entire history, $\bx_{1:t} = (\bx_1',\ldots,\bx_t')'$:
\begin{framed}
\begin{algorithm}
\label{alg:enks} \textbf{Ensemble Kalman smoother (EnKS)}\\
Start with an initial ensemble $\bx_{0|0}^{(1)},\ldots,\bx_{0|0}^{(N)} \sim \normal(\bfmu_{0|0},\bfSigma_{0|0})$. Then, at each time $t=1,\ldots,T$, do the following for each $j =1,\ldots,N$:
\begin{description}
\item[1. Forecast Step:] Compute $\bx_{t|t-1}^{(j)} = \evol_t(\bx_{t-1|t-1}^{(j)}) + \bw_t^{(j)}$, where $\bw_t^{(j)} \sim \normal_n(\bfzero,\bQ_t)$.
\item[2. Update Step:] For $l=1,\ldots,t$ compute $\bx_{l|t}^{(j)}  = \bx_{l|t-1}^{(j)} + \widehat{\bK}_{l,t}(\by_t - \widetilde\by_t^{(j)})$
where $\widetilde\by_t^{(j)}  = \lobs_t\bx_{t|t-1}^{(j)} + \bv_{t}^{(j)}$ is the pseudo-observation with $\bv_t^{(j)} \sim \normal_{m_t}(\bfzero,\bR_t)$, $\widehat{\bK}_{l,t}= \widehat{\bfSigma}_{l,t|t-1} \bH_t'(\bH_t\widehat{\bfSigma}_{t,t|t-1}\bH_t'+\bR_t)^{-1}$, and $\widehat{\bfSigma}_{l,t|t-1}$ is a regularized version of the sample cross-covariance matrix of $\bx_{l|t-1}^{(1:N)}$ and $\bx_{t|t-1}^{(1:N)}$.
\end{description}
\end{algorithm}
\end{framed}

The resulting ensemble $\bx_{t|T}^{(1:N)}$ is then an approximate sample from the smoothing distribution \textcolor{black}{$p(\bx_t|\by_{1:T})$}, and more generally, $\bx_{1:T|T}^{(1:N)}$ is an approximate sample from \textcolor{black}{$p(\bx_{1:T}|\by_{1:T})$.}

Note that regularization often results in $\widehat{\bfSigma}_{l,t|t-1}=\bfzero$ for $l=1,\ldots,t-k$ for some lag $k$. The computations in Step 2 of Algorithm \ref{alg:enks} then only have to be carried out for $k$ lags $l=t-k+1,\ldots,t$.


\section{Preliminaries \textcolor{black}{for extended ensemble Kalman techniques}} \label{sec:preliminaries}

In the previous section, we assumed that the latent $\by_{1:t}$ and the parameters $\bftheta_{1:t}$ are all known. To obtain the necessary {\color{black} background and supporting results} for the extended \textcolor{black}{EnK techniques} to be introduced in Section \ref{sec:methods}, we now extend the approaches reviewed in Section \ref{sec:enkf} to the more general HSSMs of the form \eqref{eq:transmodel}--\eqref{eq:paramodel} with non-Gaussian observations or unknown parameters.   



Assume that we have a sample $\textcolor{black}{\bftheta_{1:t-1}^{(1)},\ldots,\bftheta_{1:t-1}^{(M)}}$ from \textcolor{black}{$p(\bftheta_{1:t-1}|\bz_{1:t-1})$}, the filtering distribution of the parameters at time $t-1$. Except for special cases (see Section \ref{sec:generalspecial} below), we now need an ensemble of states, $\bx_{t-1|t-1}^{(i,1:N)}$ corresponding to each parameter sample $\textcolor{black}{\bftheta_{1:t-1}^{(i)}}$; that is, we need an {\it ensemble of ensembles}. Specifically, assume that the joint filtering distribution of parameters and states at time $t-1$ is given by
\[
\textstyle \textcolor{black}{p(\bftheta_{1:t-1},\bx_{t-1} | \bz_{1:t-1})} = \textstyle\frac{1}{MN} \sum_{i=1}^M \sum_{j=1}^N \delta_{(\bftheta_{1:t-1}^{(i)},\bx_{t-1|t-1}^{(i,j)})}(\bftheta_{1:t-1},\bx_{t-1}),
\]
the marginal filtering distribution of \textcolor{black}{$\bftheta_{1:t-1}$} is
$
\textstyle \textcolor{black}{p(\bftheta_{1:t-1} | \bz_{1:t-1}) = \textstyle\frac{1}{M} \sum_{i=1}^M \delta_{(\bftheta_{1:t-1}^{(i)})}(\bftheta_{1:t-1})},
$
and the conditional filtering distribution of $\bx_{t-1}$ given \textcolor{black}{$\bftheta_{1:t-1}^{(i)}$} is
\[
\textstyle \textcolor{black}{p(\bx_{t-1} | \bftheta_{1:t-1}^{(i)}, \bz_{1:t-1})} = \textstyle\frac{1}{N} \sum_{j=1}^N \delta_{(\bx_{t-1|t-1}^{(i,j)})}(\bx_{t-1}).
\]
\textcolor{black}{Conditional on $\bftheta_{1:t-1}^{(i)}$,} the \textcolor{black}{joint} prior distribution at time $t$ is then given by
\[
\textstyle \textcolor{black}{p(\bftheta_t,\bx_t | \bftheta_{1:t-1}^{(i)},\bz_{1:t-1}) = p(\bftheta_t| \bftheta_{t-1}^{(i)}) p(\bx_t|\bftheta_t,\bftheta_{1:t-1}^{(i)},\bz_{1:t-1}).}
\]
\textcolor{black}{Further, we have}
\begin{align}
\textcolor{black}{p(\bx_t|\bftheta_t, \bftheta_{1:t-1}^{(i)},\bz_{1:t-1})} 
& = \textstyle \int \textcolor{black}{p(\bx_t|\bx_{t-1},\bftheta_t,\bftheta_{1:t-1}^{(i)},\bz_{1:t-1})} dP(\bx_{t-1}|\bftheta_t,\bftheta_{1:t-1}^{(i)},\bz_{1:t-1}) \nonumber\\
& = \textstyle\frac{1}{N} \sum_{j=1}^N \normal_n(\bx_t|\forex_{t|t-1}^{(i,j)},\bQ_t) \approx 
 \normal_n(\bx_t|\widehat\bfmu_{t|t-1}^{(i)},\widehat\bfSigma_{t|t-1}^{(i)}), \label{eq:xprior}
\end{align}
where we have used the same EnKF approximation as in \eqref{eq:forecast2}. Note that $\widehat\bfmu_{t|t-1}^{(i)}$ and $\widehat\bfSigma_{t|t-1}^{(i)}$ are obtained as described in Section \ref{sec:tapering} above using $\bQ_t = \bQ_t(\bftheta_t)$, but only based on the $i$th forecast ensemble $\forex_{t|t-1}^{(i,1:N)}$, where $\forex_{t|t-1}^{(i,j)} = \evol_t(\bx_{t-1|t-1}^{(i,j)};\bftheta_t)$, $j=1,\ldots,N$.

Based on the approximation in \eqref{eq:xprior}, the joint distribution of all variables at time $t$ is
\begin{align}
\textcolor{black}{
p(\bz_t,\by_t,\bx_t,\bftheta_t|\bftheta_{1:t-1}^{(i)}, \bz_{1:t-1})} 
& \textcolor{black}{= p(\bz_t|\by_t, \bftheta_t)p(\by_t|\bx_t,\bftheta_t) p(\bx_t|\bftheta_t,\bftheta_{1:t-1}^{(i)},\bz_{1:t-1}) p(\bftheta_t|\bftheta_{t-1}^{(i)})} \nonumber \\
& \textcolor{black}{= f_t(\bz_t|\by_t,\bftheta_t) \normal_{m_t}(\by_t|\bH_t\bx_t,\bR_t) \normal_n(\bx_t|\widehat{\bfmu}_{t|t-1}^{(i)},\widehat{\bfSigma}_{t|t-1}^{(i)}) p(\bftheta_t|\bftheta_{t-1}^{(i)})} \label{eq:jointprior}
\end{align}
where $\bH_t = \bH_t(\bftheta_t)$ and $\bR_t=\bR_t(\bftheta_t)$ may also depend on the parameters $\bftheta_t$.

\subsection{Forecast independence \label{sec:generalspecial}}

An important special case is given by forecast independence of states and parameters \citep[][]{Frei:Kuns:12}, \textcolor{black}{where $p(\forex_t,\bftheta_t|\bz_{1:t-1}) = p(\forex_t|\bz_{1:t-1}) p(\bftheta_t|\bz_{1:t-1})$.}  Forecast independence leads to enormous \emph{simplifications} in the algorithms described later, because the prior distribution of the states can be calculated based on the entire ensemble of ensembles from the previous filtering distributions, instead of using only the corresponding ensemble as in \eqref{eq:xprior}; that is,
{\color{black}
\begin{equation}
\label{eq:forind}
\textstyle \textcolor{black}{p(\forex_t| \bftheta_t, \bz_{1:t-1}) = p(\forex_t|\bz_{1:t-1}) = \int p(\forex_t|\bx_{t-1})dP(\bx_{t-1}|\bz_{1:t-1})} = \normal_n(\forex_t|\widehat\bfmu_{t|t-1},\widetilde\bfSigma_{t|t-1}),
\end{equation}
where $\widehat\bfmu_{t|t-1}$ and $\widetilde\bfSigma_{t|t-1}$ are now based on all ensemble members $\forex_{t|t-1}^{(1:M,1:N)}$. From this forecast distribution, we can then obtain the prior distribution by adding $\bQ_t$ to the covariance matrix: $\widehat\bfSigma_{t|t-1}=\widetilde\bfSigma_{t|t-1} + \bQ_t$.
}
Because this means that we do not need a separate ensemble for each previous parameter value, we can set $N=1$ and write the forecast ensemble as $\forex_{t|t-1}^{(1:M)}$. Also, because the \textcolor{black}{forecast} distribution does not depend on $\bftheta_t$, we only have to propagate the ensemble once from time $t-1$ to $t$ in the algorithms described later, which is especially important if $\evol_t$ is computationally expensive.

Forecast independence can be proved asymptotically if the parameters $\bftheta_t$ only appear in $\bR_t$ \citep[see][]{Frei:Kuns:12}. It also arises exactly if the parameters do not appear in {\color{black}$\evol_t$} and are temporally \textcolor{black}{independent, $p(\bftheta_t|\bftheta_{t-1})=p(\bftheta_t)$}, because then
\[
{\color{black}
\textstyle\textcolor{black}{p(\forex_t| \bftheta_t, \bz_{1:t-1}) = \int p(\forex_t|\bx_{t-1},\bftheta_t)dP(\bx_{t-1}|\bftheta_t, \bz_{1:t-1}) = p(\forex_t|\bz_{1:t-1}) = \normal_n(\forex_t|\widehat\bfmu_{t|t-1},\widetilde\bfSigma_{t|t-1})},
}
\]
as  
{\color{black}in \eqref{eq:forind}. Instead of the joint distribution in \eqref{eq:jointprior}, it can then be shown that
\begin{align}
\label{eq:jointpriorind}
p(\bz_t,\by_t,\bx_t,\bftheta_t|\bz_{1:t-1})
& = p(\bz_t|\by_t, \bftheta_t)p(\by_t|\bx_t,\bftheta_t) p(\bx_t|\bftheta_t,\bz_{1:t-1}) p(\bftheta_t) \nonumber \\
& = f_t(\bz_t|\by_t,\bftheta_t) \normal_{m_t}(\by_t|\bH_t\bx_t,\bR_t) \normal_n(\bx_t|\widehat{\bfmu}_{t|t-1},\widetilde{\bfSigma}_{t|t-1}+\bQ_t) p(\bftheta_t). 
\end{align}
}

\subsection{Approximation of the likelihood for known $\by_t$ \label{sec:likapprox}}

It is also possible to obtain an approximation of the likelihood based on the EnKF. We first study likelihood approximations for the special case of known and observed $\by_t$ (i.e., the hierarchical SSM given by \eqref{eq:obsmodel}--\eqref{eq:paramodel}).

The distribution of all data up to time $t$ is then given by,
\[
\textcolor{black}{p(\by_{1:t} | \bftheta_{1:t}) \textstyle = \prod_{k=1}^t \, p(\by_k | \bftheta_{1:t}, \by_{1:k-1}) = \prod_{k=1}^t \, p(\by_k | \bftheta_{1:k}, \by_{1:k-1}),}
\]
where we have now made the conditioning on the parameters explicit. However, \textcolor{black}{in the} high-dimensional setting with potentially nonlinear evolution considered here, the \textcolor{black}{filtered likelihood} at time $t$, $\likelihood_t(\bftheta_{1:t}) \colonequals \textcolor{black}{p(\by_t | \bftheta_{1:t}, \by_{1:t-1})}$, is not available in closed form and cannot be exactly evaluated. Instead, in any parameter-inference procedure, we need to estimate this likelihood for given fixed parameter values based on the forecast ensemble $\forex_{t|t-1}^{(1:N)}$. For the remainder of Section \ref{sec:likapprox}, we focus on a single time point $t$ and assume that $\forex_{t|t-1}^{(1:N)}$ is a sample from the true forecast distribution, \textcolor{black}{$p(\forex_t|\by_{1:t-1},\bftheta_{1:t})$}. 

For Bayesian sampling-based inference, the closer the likelihood estimator is to the true likelihood, the better the performance of the parameter-inference procedure. Specifically, unbiased estimators of the true likelihood lead to so-called {\it exact approximate} procedures that can be shown to sample from the correct target distribution; but the \emph{variability} of the estimator is also crucial, in that the variance of the logarithm of the estimator of the likelihood determines the performance of the parameter inference for finite computation time \citep[][]{Andrieu2009}. Variances smaller than 2 are often recommended for satisfactory performance \citep[e.g.,][]{Doucet2015}.

\subsubsection{Particle likelihood \label{sec:pflik}}
If the \textcolor{black}{mixture} prior distribution in \eqref{eq:forecast1} were used, the filtered likelihood at time $t$ would also be a mixture of normals,
\begin{align}
\textcolor{black}{p(\by_t | \bftheta_{1:t}, \by_{1:t-1})} 
& = \textstyle\frac{1}{N} \sum_{j=1}^N \normal_{m_t}(\by_t|\bH_t \forex_{t|t-1}^{(j)},\bH_t \bQ_t \bH_t'+\bR_t) \equalscolon \plik_t(\by_t|\bftheta_t,\forex_{t|t-1}^{(1:N)}). \label{eq:pflik}
\end{align}
Except for the fact that we are able to integrate out the innovation error with covariance $\bQ_t$ here due to the linear Gaussian observations, $\plik_t$ is the likelihood approximation used in several inference procedures based on particle filters \textcolor{black}{\citep[e.g.,][]{Andrieu2010,Johansen2012}} and in some approaches in the EnKF literature \cite[e.g.,][]{UenoNaka:14,UenoNaka:16}. Using the notation from \eqref{eq:pfweights}, we can write this likelihood as $\plik_t(\by_t|\bftheta_t,\forex_{t|t-1}^{(1:N)}) \propto (1/N) \sum_{j=1}^N \alpha_t^{(j)}$. While the $\alpha_t^{(j)}$ and thus also the \emph{particle likelihood} $\plik_t$ are unbiased for the true likelihood \citep{DelMoral2004}, leading to exact approximate algorithms, the variance of the particle likelihood can be very large in high dimensions, resulting in algorithms with potentially poor performance in finite computation time. We now illustrate this using a simple example.
\begin{example}
\label{ex:ind}
For a single time point, consider the forecast distribution $\forex \sim \normal_n(\bfzero,\kappa\bI_n)$, with model error covariance $\bQ=\bfzero$, noise covariance $\bR = \theta\bI_n$, observation matrix $\bH=\bI_n$, and thus $n=m$; that is, $\by \sim \normal_n(\bfzero,(\kappa+\theta)\bI_{n})$. The unknown parameter is $\theta$.
\end{example}
In the setting of Example \ref{ex:ind}, Figure \ref{fig:likcomparison} shows the distribution of the particle likelihood $\plik_t$, which becomes increasingly skewed for increasing $n$. Further, it shows the trace plot for a ``pseudo-marginal'' Metropolis-Hastings (MH) sampler targeting $\theta$ for $n=50$ using $\plik_t$, whose large variances cause the Markov chain to get ``stuck'' after a particularly large likelihood is obtained. Finally, we explore the ensemble size $N$ necessary to keep the variance of the particle loglikelihood, which was averaged over 100 realizations of $\by$, to below 2, and this number clearly scales \emph{exponentially} with the dimension $n$. This result can also be shown analytically:
\begin{prop}
\label{prop:pflik}
In the setting of Example \ref{ex:ind}, we have $\var(\log \plik_t(\by|\theta,\forex^{(1:N)})) = \mathcal{O}(e^n/N)$ as $n,N \rightarrow \infty$, where the variance is evaluated with respect to $\forex^{(1:N)}$, and the $\forex^{(j)}$ are iid from the true forecast distribution.
\end{prop}
All proofs are provided in Appendix \ref{app:proofs}.

\begin{figure}
	\begin{subfigure}{.48\textwidth}
	\centering
	\includegraphics[width =1\linewidth]{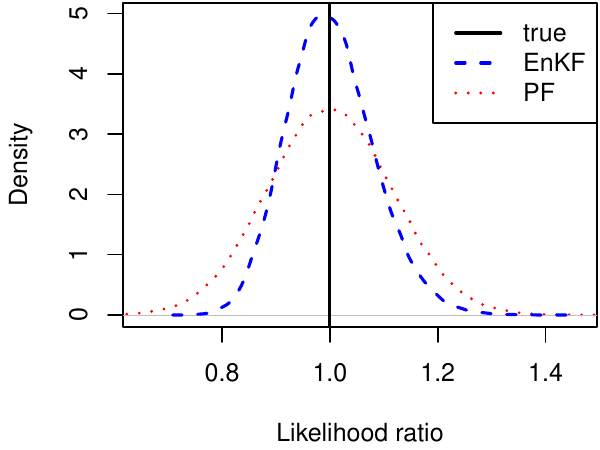}
	\caption{Distribution of likelihood \textcolor{black}{ratios} for $n=1$}
	\label{fig:likdist1}
	\end{subfigure}%
\hfill
	\begin{subfigure}{.48\textwidth}
	\centering
	\includegraphics[width =1\linewidth]{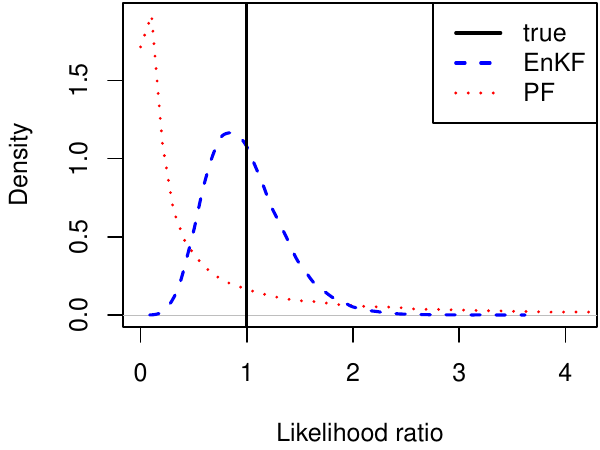}
	\caption{Distribution of likelihood \textcolor{black}{ratios} for $n=6$}
	\label{fig:likdist2}
	\end{subfigure}%

\vspace{3mm}
	\begin{subfigure}{.48\textwidth}
	\centering
	\includegraphics[width =1\linewidth]{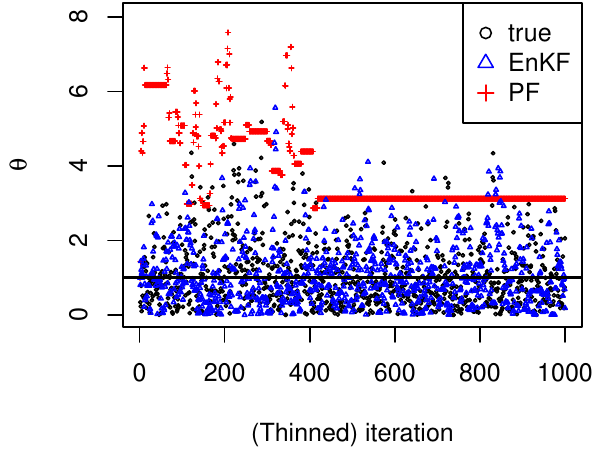}
	\caption{Trace plots for MH (thinned by factor 5)}
	\label{fig:MHtrace}
	\end{subfigure}%
\hfill
	\begin{subfigure}{.48\textwidth}
	\centering
	\includegraphics[width =1\linewidth]{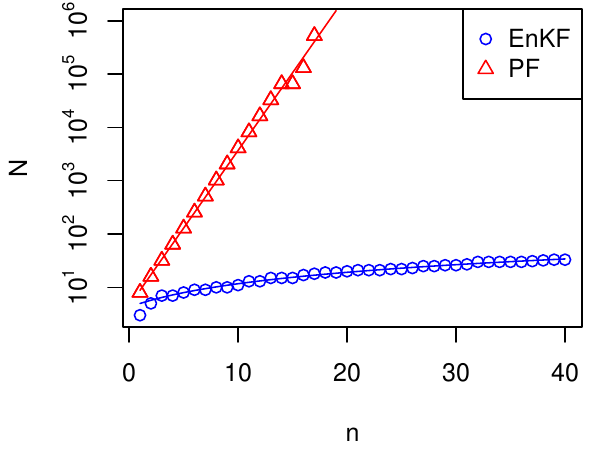}
	\caption{Necessary $N$ to keep the variance of the loglikelihood below 2}
	\label{fig:below2}
	\end{subfigure}%
%
\caption{Comparison of the likelihoods discussed in Section \ref{sec:likapprox} in the setting of Example \ref{ex:ind} with $\kappa=4$, and true parameter value $\theta=1$. Panels (\subref{fig:likdist1}) and (\subref{fig:likdist2}): Distributions of EnKF and PF likelihood approximations divided by the true likelihood at $\theta=1$ for $N=50$ with fixed data $\by$ and varying forecast ensembles. Panel (\subref{fig:MHtrace}): Trace plot for ``pseudo-marginal'' Metropolis-Hastings (MH) algorithms targeting $\theta$ with a uniform prior on $\mathbb{R}^+$, $n=N=50$, and a normal proposal distribution with standard deviation 0.8. Panel (\subref{fig:below2}): Minimum $N$ to keep the variance of the loglikelihood below 2, along with a least-squares line fitted on the original scale (for EnKF) and on a log-scale (for PF).}
\label{fig:likcomparison}
\end{figure}

The results obtained here for a particle or importance-sampling likelihood at a single time point extend directly to a particle filter in a state-space setting. While a particle filter used to approximate a likelihood in a sequential setting will typically include a resampling step at every time point \citep[e.g.,][]{Andrieu2010}, the resulting forecast or prior distribution at the next time point can be at most as good as the iid sampling from the true prior considered here, assuming that $\evol_t$ is intractable, and so the particle-filter approximation to the likelihood will also degenerate in high dimensions.

\subsubsection{EnKF likelihood \label{sec:elik}}

Due to the problems of the PF likelihood in high dimensions, we use the EnKF approximation of the prior distribution in \eqref{eq:forecast2} to obtain an \emph{EnKF likelihood}:
\begin{align}
\textcolor{black}{p(\by_t | \bftheta_{1:t}, \by_{1:t-1})} & = \textstyle \int \textcolor{black}{p(\by_t|\bx_t,\bftheta_{1:t},\by_{1:t-1})}dP(\bx_t|\bftheta_{1:t},\by_{1:t-1}) \nonumber \\
& = \textstyle\int \normal_{m_t}(\by_t|\bH_t\bx_t,\bR_t) \normal_n(\bx_t|\widehat{\bfmu}_{t|t-1},\widehat{\bfSigma}_{t|t-1}) d\bx_t \nonumber \\
& = \normal_{m_t}(\by_t|\bH_t\widehat{\bfmu}_{t|t-1},\bH_t\widehat{\bfSigma}_{t|t-1}\bH_t'+\bR_t) \equalscolon \elik_t(\by_t|\bftheta_t,\forex_{t|t-1}^{(1:N)}), \label{eq:filteringlikdef}
\end{align}
where $\bH_t$, $\bR_t$, and $\widehat{\bfmu}_{t|t-1}$ and $\widehat{\bfSigma}_{t|t-1}$ may depend on $\bftheta_t$ through $\evol_t$ and $\bQ_t$.

The EnKF likelihood has been successfully used in several examples and applications \citep[e.g.,][]{Mitc:Hout:00,Stro:Stei:Lesh:10,Frei:Kuns:12,Stroud2016}. For linear SSMs, the true forecast distribution and likelihood are both exactly Gaussian, and so the EnKF likelihood converges to the true likelihood as $N \rightarrow \infty$, assuming no localization for $\widehat{\bfSigma}_{t|t-1}$. For finite $N$, in contrast to the PF likelihood, the EnKF likelihood is a biased estimator of the true likelihood, and thus does not lead to exact approximate inference. Special care should be taken regarding parameters influenced by tapering, as, for example, a spatial-scale parameter in the innovation covariance $\bQ_t$ might be overestimated to compensate for an overly restrictive taper in the forecast covariance $\widehat{\bfSigma}_{t|t-1}$ (see Section \ref{sec:tapering}). The estimated parameter is then useful for running an EnKF with the same taper settings, but it might not be the ``true'' spatial scale of the innovation covariance in the model \eqref{eq:transmodel}--\eqref{eq:paramodel}. 

However, more importantly, the variance of the EnKF (log)likelihood can be relatively small even in high dimensions, as long as a good regularized estimator $\widehat{\bfSigma}_{t|t-1}$ can be found. This can be viewed as a manifestation of the well-known trade-off between bias and variance that is central to statistical learning with ``big'' data \citep[e.g.,][]{Hastie2009}.
To illustrate this point, we consider the EnKF likelihood in the setting of Example \ref{ex:ind} with a diagonal $\widehat{\bfSigma}_{t|t-1}$. Figure \ref{fig:likcomparison} shows that the EnKF likelihood is less skewed, less variable, and thus closer to the true likelihood than the particle likelihood. This also leads to much better mixing for the MH sampler, and the resulting estimated posterior for $\theta$ (not shown) is nearly indistinguishable from the estimated posterior using the true likelihood. Figure \ref{fig:likcomparison}(\subref{fig:below2}) indicates that for the EnKF likelihood, $N$ only has to increase \emph{linearly} with the dimension $n$ to keep the variance of the loglikelihood bounded. This can also be shown analytically:
\begin{prop}
\label{prop:enkflik}
In the setting of Example \ref{ex:ind}, we have $\var(\log \elik_t(\by|\bftheta,\forex^{(1:N)})) = \mathcal{O}(n/N)$ as $n,N \rightarrow \infty$, where the variance is evaluated with respect to $\forex^{(1:N)}$, and the $\forex^{(j)}$ are iid from the true forecast distribution.
\end{prop}
To explore whether this result holds more broadly than just in the independent case, we repeated the simulations in Figure \ref{fig:likcomparison}(\subref{fig:below2}) for a slightly modified version of Example \textcolor{black}{\ref{ex:ind}} with a tridiagonal forecast covariance matrix $\bfSigma$. This produced a plot very similar to Figure \ref{fig:likcomparison}(\subref{fig:below2}).

If the true forecast distribution, and thus also the true likelihood, is non-Gaussian, the EnKF likelihood converges to the Dawid--Sebastiani scoring rule \citep{Dawid1999} for $N \rightarrow \infty$, and so, from a frequentist perspective, parameter inference based on the EnKF likelihood at a single time point can be interpreted as approximate minimum contrast estimation \citep{Birge1993,Pfanzagl1969} under the Dawid--Sebastiani score. Due to the the propriety of this scoring rule, the minimum-contrast estimator can be shown, using the proof of \citet{Wald1949}, to be consistent under regularity conditions \citep[cf.][Sect.~3.4.2]{Gneiting2014}. From a Bayesian perspective, inference based on the EnKF likelihood essentially amounts to computing the approximate ``Gibbs posterior'' using the Dawid--Sebastiani score as the empirical risk function \citep{Jiang2008,Li2013,Katzfuss2015a}. Under certain conditions, most importantly that $\hat\bftheta_t \colonequals \argmax \elik_t(\by_t|\bftheta_t,\forex_t^{(1:N)})$ is consistent and asymptotically normal as $m_t \rightarrow \infty$, it might be possible to scale and rotate posterior samples of $\bftheta_t$ to reflect the true asymptotic uncertainty in the frequentist sense \citep{Shaby2014}, but this will be pursued in future work.

%

\subsection{Integrated likelihood \label{sec:intlikelihood}}

We now extend the EnKF likelihood for known $\by_t$ from Section \ref{sec:elik} to the full HSSM with the general transformation layer \eqref{eq:transmodel}: $\bz_t | \by_t, \bftheta_t \sim f_t(\bz_t|\by_t,\bftheta_t)$. This \emph{integrated EnKF likelihood} of the data $\bz_t$ at time $t$, with $\bx_t$ and $\by_t$ integrated out, is an important component of some of the algorithms described later. Using \eqref{eq:jointprior}, we have
\begin{align}
\textcolor{black}{p(\bz_t | \bftheta_t, \bftheta_{1:t-1}^{(i)}, \bz_{1:t-1})}
  & = \int f_t(\bz_t|\by_t,\bftheta_t) \int \normal_{m_t}(\by_t|\bH_t\bx_t,\bR_t) \textcolor{black}{p(\bx_t|\bftheta_t, \bftheta_{1:t-1}^{(i)},\bz_{1:t-1})} d\bx_t d\by_t \nonumber \\
  & = \int f_t(\bz_t|\by_t,\bftheta_t) \elik_t(\by_t|\bftheta_t,\forex_{t|t-1}^{(i,1:N)}) d\by_t \equalscolon \mlik_t(\bz_t|\bftheta_t,\forex_{t|t-1}^{(i,1:N)}), \label{eq:enkflik}
 \end{align}
where $\elik_t$ is defined in \eqref{eq:filteringlikdef}, and $\bH_t$, $\bR_t$, and $\widehat{\bfmu}_{t|t-1}^{(i)}$ and $\widehat{\bfSigma}_{t|t-1}^{(i)}$ may depend on $\bftheta_t$.

Evaluating this likelihood requires carrying out an integral over the $m_t$-dimensional vector $\by_t$. Strategies for doing so depend on the choice of transformation distribution $f_t$. If the data $\bz_t$ are Gaussian and the transformation is simply the identity, the integral degenerates and analytical integration is trivial. If it is possible to sample from \textcolor{black}{$p(\by_t|\bz_t,\bx_t,\bftheta_t)$}, the integral can be approximated using these posterior samples \citep[e.g.,][]{Chib2001,Raftery2007}. Finally, if $f_t(\bz_t|\by_t,\bftheta_t)$ is a \textcolor{black}{continuous} function in $\by_t$, a Laplace approximation \citep{Tierney1986} can be used to approximate the integral. As we are already approximating $\by_t$ as a latent Gaussian field in $\elik_t(\by_t|\bftheta_t,\forex_{t|t-1}^{(i,1:N)})$, the Laplace approximation of the integral should often result in negligible additional error, especially if the observations in $\bz_t$ are independent given $\by_t$ and $\bftheta_t$ \citep[e.g.,][]{Rue2009}; see Section \ref{sec:penkfexample} for an example.


\section{Extended ensemble Kalman \textcolor{black}{techniques} for hierarchical state-space models \label{sec:methods}}

For HSSMs of the form \eqref{eq:transmodel}--\eqref{eq:paramodel} with non-Gaussian observations or unknown parameters, we now propose extended \textcolor{black}{EnK techniques}. The general idea is to approximate the inference on the state vector in suitable existing Bayesian inference techniques using the EnKF or the EnKS.


\subsection{Extended ensemble Kalman filters\label{sec:eenkfs}}

We begin by discussing extended EnKFs that are suitable for filtering inference, which at each time $t$ aims to find the filtering distribution \textcolor{black}{$p(\bx_t,\bftheta_t|\bz_{1:t})$} given the data $\bz_{1:t}$ collected up to time $t$. Smoothing inference is discussed in Section \ref{sec:smoothing} below.

\subsubsection{Gibbs ensemble Kalman filter {\color{black}(GEnKF)}\label{sec:genkf}}

First, we consider a combination of an EnKF and a Gibbs sampler \citep{Geman1984,Gelfand1990} for filtering inference in HSSMs. For this algorithm we assume 
{\color{black}
that the parameters do not appear in $\evol_t$ and are temporally independent (i.e., \textcolor{black}{$p(\bftheta_t|\bftheta_{t-1})=p(\bftheta_t)$)}. As noted in Section \ref{sec:generalspecial}, this implies forecast independence,}
which enables inference based on a \emph{single ensemble} $\forex_{t|t}^{(1:M)} = \{\forex_{t|t}^{(1)},\ldots,\forex_{t|t}^{(M)}\}$ (i.e., $N=1$). At each time $t$, our proposed Gibbs ensemble Kalman filter (GEnKF) iteratively samples $\by_t$, $\bx_t$, and $\bftheta_t$ from their respective full conditional distributions (FCDs), which are proportional to the joint distribution in {\color{black}\eqref{eq:jointpriorind}. The state $\bx_t$ is sampled by drawing from the forecast ensemble and using the EnKF update, which allows application of the GEnKF to high-dimensional and nonlinear SSMs.}
\begin{framed}
\begin{algorithm}
\label{alg:genkf} \textbf{Gibbs ensemble Kalman filter (GEnKF)}\\
Starting with an initial ensemble $\{(\bftheta_0^{(i)},\bx_{0|0}^{(i)}): i=1,\ldots,M\}$, where $(\bftheta_0^{(i)},\bx_{0|0}^{(i)}) \sim \textcolor{black}{p(\bftheta_0,\bx_0)} = p_0(\bftheta_0)\normal_n(\bx_0|\bfmu_{0|0},\bfSigma_{0|0})$, do the following for each $t=1,2,\ldots$:
\begin{enumerate}
\item Obtain the forecast ensemble as $\forex_{t|t-1}^{(i)} = \evol_t(\bx_{t-1|t-1}^{(i)})$, for all $i=1,\ldots,M$.
\item Find starting values for $\by_t^{(i)}$ and $\bftheta_t^{(i)}$ for all $i=1,\ldots,M$.
\item For $i=1,\ldots,M$, iterate between the following steps until convergence:
 \begin{enumerate}
  \item Sample $\bx_{t|t}^{(i)}$ from \textcolor{black}{$p(\bx_t | \by_{t}^{(i)}, \bftheta_t^{(i)},\bz_{1:t})$} by \textcolor{black}{sampling $j$ uniformly from $\{1,\ldots,M\}$ and} computing the EnKF update $\bx_{t|t}^{(i)} = \update_t(\forex_{t|t-1}^{(\textcolor{black}{j})}|\by_{t}^{(i)},\forex_{t|t-1}^{(1:M)},\bftheta_{t}^{(i)})$ (see \eqref{eq:enkfupdate}).
  \item Sample $\by_{t}^{(i)}$ from $\textcolor{black}{p(\by_t | \bx_{t|t}^{(i)}, \bftheta_{t}^{(i)}, \bz_{1:t})} \propto  f_t(\bz_t|\by_t,\bftheta_t^{(i)}) \normal_{m_t}(\by_t|\bH_t(\bftheta_{t}^{(i)})\bx_{t|t}^{(i)},\bR_t(\bftheta_{t}^{(i)}))$.
  \item Sample $\bftheta_{t}^{(i)}$ from 
  $ \textcolor{black}{p(\bftheta_t | \by_{t}^{(i)},\bx_{t|t}^{(i)},\bz_{1:t})} \propto f_t(\bz_t|\by_t^{(i)},\bftheta_t) \allowbreak \normal_{m_t}(\by_t^{(i)}|\bH_t(\bftheta_t)\bx_{t|t}^{(i)},\bR_t(\bftheta_t)) \allowbreak \textcolor{black}{ \normal_n(\bx_t|\widehat{\bfmu}_{t|t-1},\widetilde{\bfSigma}_{t|t-1}+\bQ_t(\bftheta_t))} \allowbreak \textcolor{black}{p}(\bftheta_t)$.
 \end{enumerate}
 After convergence, $(\bx_{t|t}^{(i)},\bftheta_{t}^{(i)})$ is a joint sample from \textcolor{black}{$p(\bx_t,\bftheta_t|\bz_{1:t})$}, and $\bx_{t|t}^{(1:M)}$ is the filtering ensemble at time $t$.
\end{enumerate}
\end{algorithm}
\end{framed}

Note that iterating through Step 3(a)--(c) of Algorithm \ref{alg:genkf} can be done completely independently (i.e., in parallel) for $i=1,\ldots,M$. 
{\color{black} To sample $j$ in Step 3(a), we use a systematic sample across $i=1,\ldots,M$, so that each forecast ensemble member is chosen exactly once for each iteration of the Gibbs sampler. It can be shown under mild regularity conditions, that the closer the forecast ensemble is to multivariate normal, the closer the EnKF in Step 3(a) gets to sampling from the true FCD of $\bx_t$ as $M$ increases, and hence the closer the resulting Markov chain will be to the Markov chain that would be obtained using an exact Gibbs sampler, in terms of total variation distance \citep{Alquier2016}.}
As the GEnKF requires sampling from the FCDs of $\by_t$ and $\bftheta_t$ in (b) and (c), respectively, the algorithm is most efficient if the distributions $f_t(\bz_t|\by_t,\bftheta_t)$ and $p(\bftheta_t)$ are such that these FCDs are available in closed form. Otherwise, Metropolis updates are necessary for these variables, which is generally only feasible if the number of parameters in $\bftheta_t$ is \textcolor{black}{not too large}, and if the data in $\bz_t$ are independent given $\bx_t$.

For computational feasibility of Algorithm \ref{alg:genkf}, it is crucial to find good starting values for $\by_t$ and $\bftheta_t$, in order to minimize the required number of {\color{black} Gibbs} iterations. In some applications, it is natural to use the values from the previous time point (i.e., set the initial values as $\by_t^{(i)}=\by_{t-1}^{(i)}$ and $\bftheta_t^{(i)}=\bftheta_{t-1}^{(i)}$). Another possibility is to choose the value of $\bftheta_t$ that maximizes the integrated likelihood in \eqref{eq:enkflik}. 

{\color{black}
As an example, for the model with heavy-tailed data in Example \ref{ex:tdist}, the full-conditional distribution of $\bftheta_t$ is a closed-form inverse-Gamma distribution, and so we can apply the following special case of the GEnKF in Algorithm \ref{alg:genkf}:
\begin{example}[Robust ensemble Kalman filter (REnKF)]
\label{ex:renkf}
Starting with an initial ensemble $\bx_{0|0}^{(1:M)}$, where $\bx_{0|0}^{(i)} \sim \normal_n(\bx_0|\bfmu_{0|0},\bfSigma_{0|0})$, do the following for each $t=1,2,\ldots$:
\begin{enumerate}
\item Obtain the forecast ensemble as $\forex_{t|t-1}^{(i)} = \textcolor{black}{\evol_t(\bx_{t-1|t-1}^{(i)})}$, for all $i=1,\ldots,M$.
\item Set $\by_t = \bz_t$, and $\theta_{t,l}^{(i)}=1$ for all $i=1,\ldots,M$ and $l=1,\ldots,m_t$.
\item For $i=1,\ldots,M$, iterate between the following steps until convergence:
 \begin{enumerate}
  \item Sample $\bx_{t|t}^{(i)}$ from \textcolor{black}{$p(\bx_t | \by_{1:t}, \bftheta_t^{(i)})$} by sampling $j$ uniformly from $\{1,\ldots,M\}$ and computing the EnKF update $\bx_{t|t}^{(i)} = \update_t(\forex_{t|t-1}^{(j)}|\by_{t}^{(i)},\forex_{t|t-1}^{(1:M)},\bftheta_{t}^{(i)})$ (see \eqref{eq:enkfupdate}).
  \item Sample $\textstyle\theta_{t,l}^{(i)} | \by_{1:t}, \bx_{t|t}^{(i)} \stackrel{ind}{\sim} \ig\Big(\frac{\kappa+1}{2},\frac{\kappa}{2}+\Big( \frac{y_{t,l} - (\bH_t \bx_{t|t}^{(i)})_l}{\textcolor{black}{\sigma_t}} \Big)^2/2\Big)$, $l=1,\ldots,m_t$.
 \end{enumerate}
\end{enumerate}
If $\kappa$ is unknown with prior \textcolor{black}{$p(\kappa)$}, it could be sampled in a collapsed Gibbs sampling scheme prior to updating $\bftheta_t$ by sampling from \textcolor{black}{$p(\kappa|\by_t, \bx_{t|t}^{(i)}) \propto p(\kappa) \prod_{l=1}^{m_t} t_\kappa((y_{t,l} - (\bH_t \bx_{t|t}^{(i)})_l)/\sigma_t)$}.
\end{example}
Note that the iterative Huber-EnKF \citep{Rao2015} seems to be similar in spirit to our robust EnKF, but it uses the Huber norm in place of the loss function implied by a $t$-distribution.

The GEnKF can also be applied to threshold models as in Example \ref{ex:rainfall}:
\begin{example}[GEnKF for rainfall threshold model]
\label{ex:rainfallenkf}
Starting with an initial ensemble $\bx_{0|0}^{(1:M)}$, where $\bx_{0|0}^{(i)} \sim \normal_n(\bx_0|\bfmu_{0|0},\bfSigma_{0|0})$, do the following for each $t=1,2,\ldots$:
\begin{enumerate}
\item Obtain the forecast ensemble as $\forex_{t|t-1}^{(i)} = \evol_t(\bx_{t-1|t-1}^{(i)})$, for all $i=1,\ldots,M$.
\item Find starting values for $\by_t^{(i)}$ for all $i=1,\ldots,M$.
\item For $i=1,\ldots,M$, iterate between the following steps until convergence:
 \begin{enumerate}
  \item Sample $\bx_{t|t}^{(i)}$ from \textcolor{black}{$p(\bx_t | \by_{1:t}^{(i)})$} by sampling $j$ uniformly from $\{1,\ldots,M\}$ and computing the EnKF update $\bx_{t|t}^{(i)} = \update_t(\forex_{t|t-1}^{(j)}|\by_{t}^{(i)},\forex_{t|t-1}^{(1:M)})$ (see \eqref{eq:enkfupdate}).
  \item Sample $\theta_t$ from $\textcolor{black}{p(\theta_t | \bz_t,\bx_t)} \propto \textcolor{black}{p(\theta_t)} \prod_{l: z_{t,l}>0} \normal\big(z_{t,l}^{1/\theta_t}| (\bH_t \bx_t)_l,\textcolor{black}{\sigma_t^2}\big) \, z_{t,l}^{(1/\theta_t)-1}/\theta_t$. 
  \item For $l=1,\ldots,m_t$, sample from $\textcolor{black}{p(y_{t,l}|z_{t,l},\bx_t,\theta_t)} = \begin{cases} \delta_{z_{t,l}^{1/\theta_t}}(y_{t,l}), &  z_{t,l}>0, \\ \normal^{-}\big((\bH_t \bx_t^{(i)})_l,\sigma_t^2\big), &  z_{t,l}=0, \end{cases}$ \\
  where $\normal^{-}$ denotes a truncated normal distribution on the negative real line.
 \end{enumerate}
\end{enumerate}
\end{example}
}


A Gibbs ensemble Kalman filter could in principle also be applied if forecast independence does not hold. This is difficult, however, because
$
\textcolor{black}{p(\bftheta_t, \bx_t | \bz_{1:t}) = \int p(\bftheta_t, \bx_t | \bftheta_{t-1}, \bz_{1:t}) dP(\bftheta_{t-1}|\bz_{1:t}),}
$
where it is straightforward to sample from \textcolor{black}{$p(\bftheta_t, \bx_t | \bftheta_{t-1}, \bz_{1:t})$} using a Gibbs sampler \textcolor{black}{similar to} above, but it is highly challenging to compute the probability or weight $P(\bftheta_{t-1}=\bftheta_{t-1}^{(i)} |\bz_{1:t})$, which will also degenerate in high dimensions. Thus we will not pursue this general case here.

Perhaps the approaches in the literature most closely related to our GEnKF are \citet{Myrseth2010} and \citet{Tsyrulnikov2015}, who model the forecast covariance matrix as random with an inverse-Wishart prior. This matrix is not a model parameter, which makes Bayesian inference difficult. \citet{Myrseth2010} propose a Gibbs-like algorithm with a single iteration, but they do not take the data into account when calculating the posterior of the forecast covariance. 
\citet{Tsyrulnikov2015} seem to ignore the fact that forecast independence does not hold in their model.


\subsubsection{Particle ensemble Kalman filter {\color{black}(PEnKF)} \label{sec:penkf}}

For high-dimensional SSMs, particle filters typically degenerate \citep[e.g.,][]{Doucet2001,Snyd:Beng:Bick:08}. While Rao-Blackwellized particle filters \citep[RBPFs;][]{Doucet2000} can scale to higher dimensions by integrating out part of the state analytically, for example using a Kalman filter \citep{Chen2000}, the integral is generally not available for HSSMs with nonlinear evolution or non-Gaussian data. An exact approximation of RBPFs proposed in \citet{Johansen2012} uses ``local'' particle filters to carry out the integration, thus employing the particle likelihood from Section \ref{sec:pflik}, which was shown to scale poorly as well.

As an alternative, we propose here a \emph{particle ensemble Kalman filter (PEnKF)}, which uses the EnKF likelihood from Sections \ref{sec:elik} and \ref{sec:intlikelihood} to approximately integrate out the high-dimensional state. The PEnKF is suitable for HSSMs for which forecast independence does not hold, or for which efficient sampling from the FCDs in Step 3 of Algorithm \ref{alg:genkf} is not possible. As a particle filter is applied to $\bftheta_t$, the PEnKF will work best if the number of unknown parameters is not too high.  Although this algorithm does not require forecast independence for the parameters as the GEnKF above, forecast independence does result in considerable simplifications (see below).

\begin{framed}
\begin{algorithm}
\label{alg:penkf} \textbf{Particle ensemble Kalman filter (PEnKF)}\\
Initialize the algorithm with an equally weighted ensemble $(\bftheta_0^{(i)},\bx_{0|0}^{(i,j)}) \sim \textcolor{black}{p(\bftheta_0,\bx_0)}= \textcolor{black}{p}(\bftheta_0)\normal_n(\bx_0|\bfmu_{0|0},\bfSigma_{0|0})$ with $w_0^{(i)}=1/M$, $i=1,\ldots,M; j=1,\ldots,N$. Then, for $t=1,2,\ldots$:
\begin{enumerate}
\item For $i=1,\ldots,M$:
 \begin{enumerate}
  \item Sample a particle $\bftheta_{t}^{(i)}$ from a proposal distribution $q_t(\bftheta_t|\bftheta_{t-1}^{(i)})$.
  \item Propagate the corresponding ensemble: $\forex_{t|t-1}^{(i,j)} = \evol_t(\bx_{t-1|t-1}^{(i,j)};\bftheta_t^{(i)})$, $j=1,\ldots,N$.
 \item Calculate the particle weight:\\
 $w^{(i)}_t \propto w^{(i)}_{t-1} \mlik_t(\bz_t|\bftheta_t^{(i)},\forex_{t|t-1}^{(i,1:N)}) p(\bftheta_{t}^{(i)} |\bftheta_{t-1}^{(i)})/ q_t(\bftheta_{t}^{(i)} |\bftheta_{t-1}^{(i)})$.
  \item Generate $\bx_{t|t}^{(i,j)}$ from \textcolor{black}{$p(\bx_t|\bz_{1:t},\bftheta_{1:t}^{(i)})$} for $j=1,\ldots,N$ (see below).
 \end{enumerate}
\item The filtering distribution is given by\\
 $\textcolor{black}{p(\bftheta_t,\bx_t| \bz_{1:t})} \approx \sum_{i=1}^M w_t^{(i)} \textstyle\frac{1}{N} \sum_{j=1}^{N} \delta_{(\bftheta_{t}^{(i)},\bx_{t|t}^{(i,j)})}(\bftheta_t,\bx_t)$.
\item If desired, resample $(\bftheta_{t}^{(i)},\bx_{t|t}^{(i,1:N)})$ to obtain an equally weighted sample \citep[see, e.g.,][for a comparison of resampling schemes]{Douc2005}. 
\end{enumerate}
\end{algorithm}
\end{framed}
By noting that \textcolor{black}{$p(\bftheta_{1:t},\bz_{1:t})=p(\bftheta_{t},\bz_{t}|\bftheta_{1:t-1},\bz_{1:t-1})p(\bftheta_{1:t-1},\bz_{1:t-1})$}, the expression for the weights in Step 1(c) is derived as follows:
\[
w^{(i)}_t \propto \frac{w^{(i)}_{t-1}}{q_t(\bftheta_{t}^{(i)} |\bftheta_{t-1}^{(i)})} \frac{\textcolor{black}{p(\bftheta_{1:t}^{(i)},\bz_{1:t})}}{\textcolor{black}{p(\bftheta_{1:t-1}^{(i)},\bz_{1:t-1})}}  = \frac{w^{(i)}_{t-1}}{q_t(\bftheta_{t}^{(i)} |\bftheta_{t-1}^{(i)})} \textcolor{black}{p(\bftheta_{t}^{(i)},\bz_{t}|\bftheta_{1:t-1}^{(i)},\bz_{1:t-1})},
\]
where $\textcolor{black}{p(\bftheta_{t},\bz_{t}|\bftheta_{1:t-1},\bz_{1:t-1})=p(\bftheta_t|\bftheta_{1:t-1},\bz_{1:t-1})p(\bz_{t}|\bftheta_{1:t},\bz_{1:t-1})=p(\bftheta_{t} |\bftheta_{t-1}) p(\bz_{t}|\bftheta_{1:t},\bz_{1:t-1})}$, and $\textcolor{black}{p(\bz_{t}|\bftheta_{1:t}^{(i)},\bz_{1:t-1})}$ is \textcolor{black}{approximated by} the integrated likelihood in \eqref{eq:enkflik}.

Thus, the PEnKF is similar to the mixture KF \citep{Chen2000} and a Rao-Blackwellized particle filter, except that the high-dimensional state $\bx$ is integrated out and sampled using the EnKF approximation. This again allows application to high-dimensional and nonlinear SSMs.

{\color{black}
In the numerical examples in Section \ref{sec:examples}, for simplicity we set the proposal distributions equal to the prior distributions implied by the model: $q_t(\bftheta_{t} |\bftheta_{t-1}) = p(\bftheta_{t} |\bftheta_{t-1})$. Of course, it might be possible to adapt existing, more sophisticated strategies for choosing efficient proposal distributions to our setting.
}

For Step 2 in Algorithm \ref{alg:penkf}, a sample from \textcolor{black}{$p(\bx_t|\bz_{1:t},\bftheta_{1:t}^{(i)})$} needs to be computed for each particle with significantly nonzero weight $w_t^{(i)}$. This sample can be obtained using a special case of the Gibbs sampler in Algorithm \ref{alg:genkf} with $\bftheta_t=\bftheta_t^{(i)}$ held fixed (i.e., skipping Step 3(c)). 
In settings for which the Laplace approximation works well (see Section \ref{sec:intlikelihood}), this Gibbs sampler could consist of only a single iteration, in which $\by_t$ is first sampled from the Gaussian approximation to its FCD used in the Laplace approximation, and then $\bx_t$ is sampled as in Step 3(a) of Algorithm \ref{alg:genkf}. 
{\color{black} Here is an example of a PEnKF using this Laplace approach for the Poisson model in Example \ref{ex:poisson}:
\begin{example}[PEnKF for dynamic Poisson model]
\label{ex:poissonpenkf}
Initialize $(\bftheta_0^{(i)},\bx_{0|0}^{(i,j)}) \sim p(\bftheta_0)\normal_n(\bx_0|\bfmu_{0|0},\bfSigma_{0|0})$ and $w_0^{(i)}=1/M$, $i=1,\ldots,M; j=1,\ldots,N$. Then, for $t=1,2,\ldots$, and each $i=1,\ldots,M$:
\begin{enumerate}
\item Sample $\bftheta_{t}^{(i)}$ from $\normal_6(\bftheta_{t-1}^{(i)},.05^2\bI_6)$.
\item Propagate the corresponding ensemble: $\forex_{t|t-1}^{(i,j)} = \levol_t(\bftheta_{t}^{(i)}) \, \bx_{t-1|t-1}^{(i,j)}$, $j=1,\ldots,N$.
\item Using Newton-Raphson, find the mode $\bfnu_t^{(i)}$, and the negative inverse Hessian $\bfLambda_t^{(i)}$ at the mode, of the function $\log g_t^{(i)}(\by)$, where $g_t^{(i)}(\by) = \elik_t(\by|\bftheta_t^{(i)},\forex_{t|t-1}^{(i,1:N)}) \prod_{l=1}^{m_t} \mathcal{P}ois(z_{t,l} | e^{y_{l}})$.
 \item Calculate the particle weight $w^{(i)}_t \propto w^{(i)}_{t-1} \mlik_t(\bz_t|\bftheta_t^{(i)},\forex_{t|t-1}^{(i,1:N)})$, where $\mlik_t(\bz_t|\bftheta_t^{(i)},\forex_{t|t-1}^{(i,1:N)}) =$ \\ $g_t^{(i)}(\bfnu_t^{(i)})/ \normal_{m_t}(\bfnu_t^{(i)}|\bfnu_t^{(i)},\bfLambda_t^{(i)})$.
\item Ensemble update: $\bx_{t|t}^{(i,j)} = \update_t(\forex_{t|t-1}^{(i,j)}|\by_{t}^{(i)},\forex_{t|t-1}^{(i,1:N)},\bftheta_{t}^{(i)})$, $j=1,\ldots,N$, where $\by_{t}^{(i)} \sim \normal(\bfnu_t^{(i)},\bfLambda_t^{(i)})$.
\item Resample $(\bftheta_{t}^{(i)},\bx_{t|t}^{(i,1:N)})$ with probability proportional to $w^{(i)}_t$; then set $w^{(i)}_t=1/M$.
\end{enumerate}
\end{example}
}


%
%

Algorithm \ref{alg:penkf} is, in principle, also applicable when $\bftheta_t \equiv \bftheta$ is constant over time (i.e., the prior distribution is a point mass at the previous value). However, in this case the proposal distribution also has to be a point mass, which means that over time there will be fewer and fewer unique particles. To avoid this, the particles can be resampled based on a kernel density estimate of the filtering distribution \citep{LiuWest:01,Frei:Kuns:12}, but then the state ensemble has to be propagated and updated again based on the newly sampled parameter values.

As described in Section \ref{sec:generalspecial}, in the case of \emph{forecast independence}, the likelihood can be computed based on the entire filtering ensemble from $t-1$, and so we can set $N=1$ and carry out Algorithm \ref{alg:penkf} with a \emph{single ensemble} instead of an ensemble of ensembles. This important special case is described in Algorithm 1 of \citet{Frei:Kuns:12}. If $\evol_t$ does not depend on $\bftheta_t$ but is computationally expensive, it can then be advantageous to ``oversample'' the parameters, so that there are more parameter particles than state ensemble members.

%
%

\subsection{\textcolor{black}{\textcolor{black}{Smoothing using} extended ensemble Kalman techniques} \label{sec:smoothing}}

We now discuss extended EnKS methods for HSSMs when smoothing inference is desired; that is, when interest is in \textcolor{black}{$p(\bx_{1:T},\bftheta_{1:T} | \bz_{1:T})$} for some fixed time period $\{1,\ldots,T\}$.

\subsubsection{Smoothing inference for high-dimensional parameter vectors}

If the parameters and $\by_t$ are high-dimensional but they can be efficiently sampled from their full-conditional distributions (FCDs), a Gibbs-EnKS can be applied:
\begin{framed}
\begin{algorithm}
\label{alg:genks} \textbf{Gibbs ensemble Kalman smoother (GEnKS)}
\begin{enumerate}
\item Initialize {\color{black}$\by_{1:T}^{(0)}$} and {\color{black}$\bftheta_{1:T}^{(0)}$}.
\item Iterate between following steps {\color{black}for $i=1,2,\ldots$} until convergence:
\begin{enumerate}
\item Obtain samples ${\color{black}\bx_{1:T|T}^{(i,1:N)}}$ from $\textcolor{black}{p(\bx_{1:T}|\bftheta_{1:T}^{(i-1)},\by_{1:T}^{(i-1)})}$ using the EnKS (Algorithm \ref{alg:enks}), and set ${\color{black}\bx_{1:T|T}^{(i)} = \bx_{1:T|T}^{(i,j)}}$ for a $j$ sampled uniformly at random from $\{1,\ldots,N\}$.
\item Draw a sample {\color{black}$\by_{1:T}^{(i)}$ from $\textcolor{black}{p(\by_{1:T}|\bz_{1:T},\bx_{1:T|T}^{(i)},\bftheta_{1:T}^{(i-1)})}$}.
\item Draw a sample {\color{black}$\bftheta_{1:T}^{(i)}$ from $\textcolor{black}{p(\bftheta_{1:T}|\bx_{1:T|T}^{(i)},\by_{1:T}^{(i)},\bz_{1:T})}$}.
\end{enumerate}
\end{enumerate}
\end{algorithm}
\end{framed}


The GEnKS can be interpreted as an alternative version of the Gibbs-KF of \citet{Carter1994} or the particle-Gibbs algorithm of \citet{Andrieu2010}, where the Kalman or particle filter, respectively, is replaced by the EnKF. This allows the GEnKS to be applied to high-dimensional, nonlinear SSMs. 
{\color{black} In Section \ref{sec:genksex}, we conduct a numerical comparison using the Lorenz-96 model from Example \ref{ex:lorenz} between particle Gibbs and the following GEnKS:
\begin{example}[GEnKS for Lorenz-96 with static parameter]
\label{ex:lorenzgenks}
Initialize $\theta^{(0)}$, set $\by_{1:T} = \bz_{1:T}$, and then iterate between following steps for $i=1,2,\ldots$ until convergence:
\begin{enumerate}
\item Obtain samples $\bx_{1:T|T}^{(i,1:N)}$ from $p(\bx_{1:T}|\theta^{(i-1)},\by_{1:T})$ using the EnKS (Algorithm \ref{alg:enks}), and set $\bx_{1:T|T}^{(i)} = \bx_{1:T|T}^{(i,j)}$ for a $j$ sampled uniformly from $\{1,\ldots,N\}$.
\item With $\forex_{t|T}^{(i)} = \lor_{8, 0.2}(\bx_{t-1|T}^{(i)})$, draw a sample $\theta^{(i)}$ from 
\begin{equation}
\label{eq:lorenzfcd}
\textstyle\normal\Big(\frac{\mu_\theta/\sigma_\theta^2+\sum_{t=2}^T \forex_{t|T}^{(i)}{}'\bQ_t^{-1}\bx_{t-1|T}^{(i)}}{1/\sigma_\theta^2+\sum_{t=2}^T \forex_{t|T}^{(i)}{}'\bQ_t^{-1}\forex_{t|T}^{(i)}},\frac{1}{1/\sigma_\theta^2+\sum_{t=2}^T \forex_{t|T}^{(i)}{}'\bQ_t^{-1}\forex_{t|T}^{(i)}}\Big).
\end{equation}
\end{enumerate}
\end{example}
}

\subsubsection{Smoothing inference for low-dimensional parameters}

If $\bftheta_{1:T}$ cannot be sampled efficiently from its FCD in Step 1(c) of Algorithm \ref{alg:genks} but $\bftheta_t$ is low-dimensional, the following smoothing algorithms could be applied.

\paragraph{Static parameters}

First, consider the case of a static (i.e., temporally constant) low-dimensional parameter vector. In this case, it is possible to apply a \emph{Metropolis-Hastings-EnKS (MHEnKS)}. This algorithm consists of proposing a new parameter value, $\bftheta^P$ say, and then accepting this value with a probability depending on the likelihood
\begin{equation}
\label{eq:smoothinglik}
\textcolor{black}{\textstyle p(\bz_{1:T} | \bftheta) = \prod_{t=1}^T p(\bz_t|\bz_{1:t-1},\bftheta) = \prod_{t=1}^T \mlik_t(\bz_t|\bftheta,\forex_{t|t-1}^{(1:N)})}
\end{equation}
for $\bftheta=\bftheta^P$, where the terms in the product on the right-hand side are the integrated likelihoods at time $t$ in \eqref{eq:enkflik}, which can be obtained sequentially by carrying out an EnKF and computing the integrated likelihood at each time $t$. Then, we can obtain samples from \textcolor{black}{$p(\bx_{1:T}|\bz_{1:T},\bftheta^{(i)})$} for each thinned posterior sample $\bftheta^{(i)}$ using a special case of Algorithm \ref{alg:genks}, in which we hold $\bftheta_t=\bftheta^{(i)}$ fixed and skip Step 2(c). If the evolution is linear, it is also possible to reuse the filtering ensemble $\bx_{T|T}^{(1:N)}$ already obtained in the EnKF by applying the EnKS from \citet{Stro:Stei:Lesh:10}.

The MHEnKS can be interpreted as an alternative version of the particle marginal Metropolis-Hastings algorithm of \citet{Andrieu2010}, where the particle likelihood in \eqref{eq:pflik} is replaced by the EnKF likelihood in \eqref{eq:filteringlikdef}. The MHEnKS should hence scale to higher state dimensions. 


\paragraph{Time-varying parameters}

If $\bftheta_t$ varies over time, inference becomes more challenging, in that even if each $\bftheta_t$ consists of a small number of elements, the combined vector of unknown parameters at all time points, $\bftheta_{1:T}$, can still be rather high-dimensional if $T$ is large. Hence, for a low-dimensional time-varying parameter vector, we propose the following particle-EnKS:
\begin{framed}
\begin{algorithm}
\label{alg:penks} \textbf{Particle ensemble Kalman smoother (PEnKS)}
\begin{enumerate}
\item Obtain smoothing trajectories $\bftheta_{1:T}^{(1:M)}$ and corresponding weights $w_{T}^{(1:M)}$ by running the PEnKF in Algorithm \ref{alg:penkf} up to time $T$, and saving and potentially resampling the trajectories $(\bftheta_{1:t}^{(i)},\bx_{1:t|t}^{(i,1:N)})$ at each time $t$ \citep{Kita:96}.
\item If desired, improve degeneracy by obtaining new state trajectories $\bx_{1:T|T}^{(i,1:N)}$ from \textcolor{black}{$p(\bx_{1:T}|\by_{1:T},\bftheta_{1:T}^{(i)})$} for each particle trajectory $\bftheta_{1:T}^{(i)}$ with significantly nonzero weight $w_T^{(i)}$, using a special case of Algorithm \ref{alg:genks}, in which we hold $\bftheta_{1:T}=\bftheta_{1:T}^{(i)}$ fixed and skip Step 2(c).
\end{enumerate}
\end{algorithm}
\end{framed}
Note that this algorithm might perform poorly when $T$ is large, because of degeneracy problems at ``early'' time points \citep[e.g.,][]{Doucet2000a}. For Step 1, we can obtain the same simplifications in the case of forecast independence as for the ``stand-alone'' PEnKF described in Section \ref{sec:penkf}. For Step 2, we can re-use state ensembles in the case of linear evolution using the EnKS procedure of \citet{Stro:Stei:Lesh:10}.

\subsubsection{A more general Markov chain Monte Carlo -- \textcolor{black}{ensemble Kalman smoother} \label{sec:mcmc}}

An extension of the Gibbs-EnKS in Algorithm \ref{alg:genks}, which we will call an MCMC--EnKS, is possible as long as $\bftheta_t$ can be split into a set of parameters $\bftheta_{1:T} = (\bftheta_{1:T}^{(1)}{}',\ldots,\bftheta_{1:T}^{(J)}{}')'$, where each $\bftheta_{1:T}^{(j)}$ is itself either low-dimensional or has a closed-form full-conditional distribution, and can be updated as in the Metropolis-Hastings-, Gibbs-, or particle-EnKS. This allows the fitting HSSMs with a variety of unknown parameters.


\subsection{Properties of the extended \textcolor{black}{ensemble Kalman} algorithms}

\subsubsection{Convergence}

{\color{black}
\begin{prop}
\label{prop:convergence}
If the evolution operators $\evol_t$ are linear and $\widehat{\bfSigma}_{t|t-1}$ is a consistent estimator of the forecast covariance matrix, the empirical distributions for the states and parameters produced by the PEnKF, GEnKS, MHEnKS, and PEnKS all converge to the true filtering or smoothing distributions under mild regularity conditions, as $N$, $M$, and the number of MCMC iterations all tend to infinity.
\end{prop}
Our informal proof is derivative of existing results and provided in Appendix \ref{app:proofs}. Note that $\widehat{\bfSigma}_{t|t-1}$ is consistent if the tapering range in Section \ref{sec:tapering} tends to infinity as $N \rightarrow \infty$.
The convergence in Proposition \ref{prop:convergence} does not hold exactly in the case of the GEnKF, in that the marginal filtering distribution of $\bx_t$ integrated over $\bftheta_t$ and $\by_t$ and the forecast distribution of $\bx_{t+1}$ are not exactly Gaussian anymore, but the GEnKF treats them as such. 
}

For SSMs with nonlinear evolution, the EnKF essentially converges to the best filter among those which update the state linearly based on the data at each time point \citep[e.g.,][]{LeGland2011,Law2014}. 
{\color{black}This is similar to the notions in so-called linear Bayes estimation that considers linear estimators given only the  first and second moments of the prior distribution and likelihood \citep[e.g,][]{hartigan1969linear,goldstein2007bayes}.  In addition, there is a close connection between approximate Bayes computation (ABC) and linear Bayes estimation \citep[e.g.,][]{Nott2014} and the EnKF can be shown to be equivalent to special cases of ABC algorithms \citep{Nott2012}. 
}
As argued in Section \ref{sec:elik}, the EnKF likelihood converges to the Dawid-Sebastiani score, which is a proper scoring rule. Thus, for nonlinear evolution the extended EnKF and EnKS techniques will not produce exact inference, but the resulting samples of the state and parameters can still provide useful approximations to the desired distributions in many settings.

For the high-dimensional applications of interest here, only small ensembles are computationally feasible, and so the algorithms' performance for \emph{small} $N$ and $M$ is much more important than their asymptotic properties for $N$ and $M$ tending to infinity. Our algorithms will tend to perform well for HSSMs for which the embedded SSM \eqref{eq:obsmodel}--\eqref{eq:evomodel} is well suited for inference using the EnKF and EnKS. This is the case for many real-world applications \citep[e.g.,][]{Houtekamer2005,Bonavita2010,Houtekamer2014}. We have also seen in Section \ref{sec:likapprox}, that the EnKF approximation to the likelihood can be much less variable than the particle approximation, which should lead to better algorithms. Further numerical comparisons are provided in Section \ref{sec:examples} below.

\subsubsection{Computational cost}

The EnKF requires storing, propagating, and operating on $N$ vectors (i.e., ensemble members) of length $n$, and the computational complexity of each EnKF update is $\mathcal{O}(nN^2)$ for most EnKF variants \citep[e.g.,][]{Tipp:Ande:Bish:03}. The computations at each time step are highly parallelizable.

For our extended EnKF algorithms, the EnKF generally has to be carried out several times. However, often only a small number of iterations or particles is necessary to obtain a reasonable approximation to the desired distribution (see Section \ref{sec:examples} below).  {\color{black} Note that in machine learning, often only a few or even a single update (called contrastive divergence)  is used to sample from restricted Boltzman machine Markov random fields for training deep neural networks \citep{Hinton2002}.  Although this adds bias, empirical and some theoretical support suggests that it is an effective and efficient approach in the deep learning context \citep[e.g.,][]{Bengio2009}.}  In addition, in the case of forecast independence for the GEnKF and PEnKF, only the update step has to be performed several times, while the forecast step only has to be carried out for a single ensemble. Thus, considering that the most expensive part of the EnKF is often the forward propagation using an expensive evolution model $\evol_t$, the increased computational cost of the GEnKF or PEnKF relative to the EnKF can be negligible.

{\color{black} 
Thus, the EnKF and our proposed extensions are \emph{scalable} to very high state dimensions $n$ and large numbers of observations $m_t$. Parameters $\bftheta_t$ that cannot be included in the state can also be high-dimensional if closed-form full-conditional distributions are available, so that our Gibbs ensemble Kalman techniques can be applied. 
However, operational, large-scale implementation for millions of dimensions or more is a significant undertaking. While efficient and parallel existing EnKF implementations such as the data assimilation research testbed, or DART \citep{Anderson2009a}, should be suitable for the proposed extensions, the application of EnKF in operational, large-scale environments typically requires serious high-performance-computing expertise and extensive tuning of inflation and localization parameters.
}

\section{Examples and numerical comparisons \label{sec:examples}}

The extended \textcolor{black}{EnK} methods described in Section \ref{sec:methods} are very general and can be applied in many situations. {\color{black}We now present numerical comparisons for the examples in Section \ref{sec:allexamples}.}


{\color{black}
\subsection{Simulation study for non-Gaussian observations \label{sec:nongaussobs}}

\begin{figure}[!tb]
\centering
	\begin{subfigure}{1\textwidth}
\centering
\includegraphics[width=.8\linewidth]{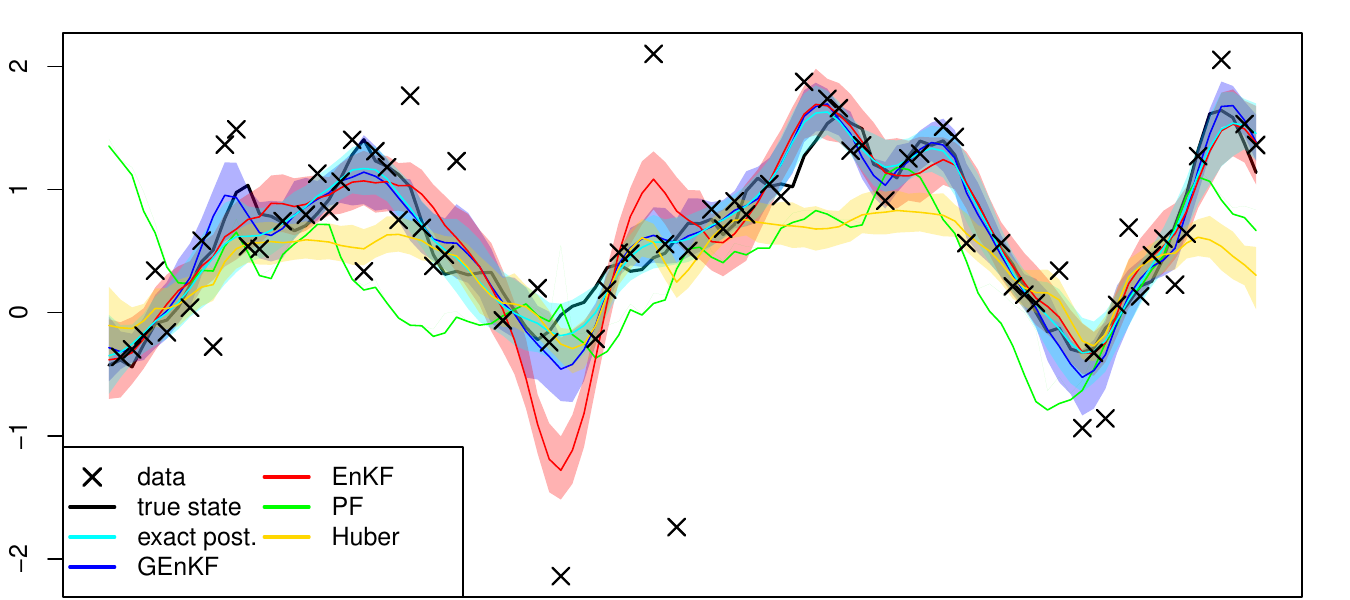} 
\caption{{\color{black}Simulated} heavy-tailed data with $t_2$ errors}
\label{fig:robust}
	\end{subfigure} \\
	\vspace{1mm}
	\begin{subfigure}{1\textwidth}
\centering
\includegraphics[width=.8\linewidth]{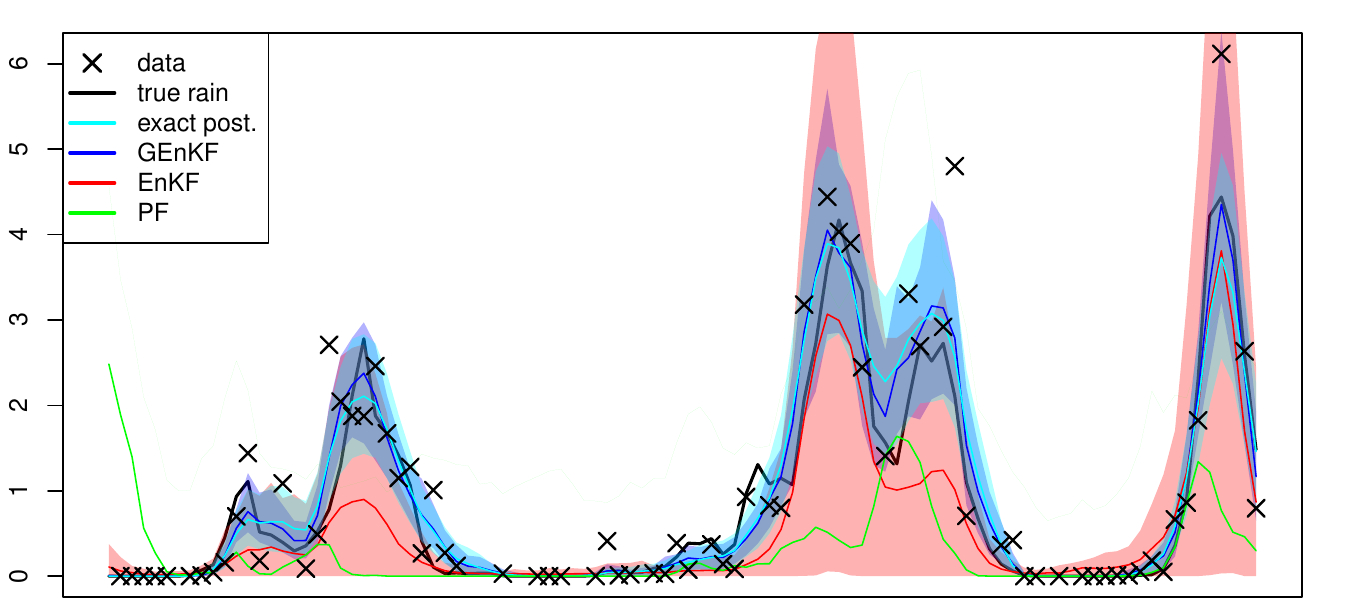}
\caption{{\color{black}Simulated} rainfall amounts with known \textcolor{black}{$\theta=3$}}
\label{fig:rain}
	\end{subfigure}%
\caption{
{\color{black}
Simulated non-Gaussian data on a one-dimensional spatial domain, together with posterior means (solid lines) and pointwise posterior 80\% credible intervals (shaded bands) as obtained by various methods
}
}
\label{fig:nongaussian}
\end{figure}

First, we considered three non-Gaussian observation scenarios with temporally independent parameters with closed-form full conditionals, which can be handled using our GEnKF (Algorithm \ref{alg:genkf}). In particular, we consider \emph{heavy-tailed} observation noise from a $t$ distribution with $\kappa=2$ degrees of freedom from Example \ref{ex:tdist} (see Figure \ref{fig:robust}), and the \emph{rainfall} threshold model from Example \ref{ex:rainfall} with the true $\kappa=3$ assumed known (see Figure \ref{fig:rain}) or unknown. 

We implemented the special cases of the GEnKF presented in Examples \ref{ex:renkf} and \ref{ex:rainfallenkf} for the heavy-tailed and rainfall data, respectively. In the case of known $\theta_t$ for the rainfall data, we simply skipped Step 3(b) in Example \ref{ex:rainfallenkf}. 
The GEnKFs were compared to a standard \textcolor{black}{sequential importance resampling (SIR)} particle filter, and to the true posterior distribution obtained using an MCMC that was run for 1,000 iterations after 500 burn-in iterations.
We also considered a commonly used extension of the EnKF update \citep{Ande:03}: For an observation equation of the form $\by_t = \obs_t(\bx_{t},\bv_t)$, the pseudo-observation in \eqref{eq:enkfupdate} is then replaced by 
$
\widetilde\by_t^{(j)} = \obs_t(\bx_{t|t-1}^{(j)},\bv_t^{(j)}),
$
and the Kalman gain is estimated as $\widehat\bK_t = \widehat\bC_{xy} \widehat\bC_{yy}^{-1}$, where $\widehat\bC_{yy}$ and $\widehat\bC_{xy}$ are regularized versions of the empirical covariance matrix of $\widetilde\by_t^{(1:N)}$ and the cross-covariance matrix between $\bx_{t|t-1}^{(1:N)}$ and $\widetilde\by_t^{(1:N)}$, respectively.
For the rainfall data with unknown \textcolor{black}{$\theta$}, we used state augmentation.
For the heavy-tailed data, we also implemented the robust Huber-EnKF proposed in \citet{Roh2013}. The Huber-EnKF requires a cutoff or clipping value for each observation location, which we selected by sampling 1000 times from the true forecast distribution based on an update consisting of a single observation at the respective location.

Because non-Gaussian observations are not directly related to the evolution model and the forecast step, we focused on the update step at a single time point $t$. We simulated 100 true state vectors of size $n=100$ on a one-dimensional spatial domain $[1,100]$ with $m_t=75$ randomly chosen observation locations. The prior of the state was taken to be a normal distribution with known true mean 0.2 (constant over space) and known true covariance matrix based on a powered exponential covariance function with power 1.8 and scale parameter 10, same for all three simulation scenarios. Also, we assumed \textcolor{black}{$\sigma_t = 0.2$} to be known for all examples. 

}

For the standard EnKF, the Huber-EnKF, and the particle filter (PF), we used the same prior ensemble of size $M=100$, simulated from the true prior distribution. For our GEnKF we only used a subset of size $M=30$ of this ensemble, and only one Gibbs iteration for rainfall with known \textcolor{black}{$\theta$}, and three iterations for the other two scenarios. For the EnKF and GEnKF, we used the same Wendland taper function with tapering length 20.

\begin{table}[t]
\centering
{\color{black}
\begin{tabular}{r|rr|rr|rr}
 & \multicolumn{2}{c|}{Heavy-tailed} 
 & \multicolumn{2}{c|}{~~Rainfall (\textcolor{black}{$\theta=3$}) ~~} 
 & \multicolumn{2}{c}{Rainfall (\textcolor{black}{$\theta$} unknown)} \\ 
 & RMSPE & CRPS & RMSPE & CRPS & RMSPE & CRPS \\ 
  \hline
exact & 0.167 & 0.093 & 0.435 & 0.144 & 0.422 & 0.131 \\ 
  GEnKF & \textbf{0.185} & \textbf{0.103} & \textbf{0.452} & \textbf{0.152} & \textbf{0.939} & \textbf{0.220} \\ 
  EnKF & 0.288 & 0.149 & 26.714 & 7.427 & $>$100 & $>$100 \\ 
  PF & 0.767 & 0.586 & 2.031 & 1.042 & 3.912 & 1.491 \\ 
  Huber & 0.643 & 0.414 &  &  &  &  \\ 
\end{tabular}
\caption{\label{tab:sim}Simulation results for non-Gaussian observations: {\color{black} heavy-tailed data (Example \ref{ex:tdist}) and rainfall data (Example \ref{ex:rainfall}).} MSPE = mean squared prediction error; CRPS = continuous ranked probability score; exact = posterior distribution obtained using MCMC; PF = \textcolor{black}{SIR particle filter}; Huber = robust Huber-EnKF from \citet{Roh2013}}
}
\end{table}

The results are shown in Figure \ref{fig:nongaussian} and Table \ref{tab:sim}. We computed both the mean squared prediction error (MSPE), which evaluates the accuracy of the \textcolor{black}{posterior mean as a point prediction}, and the continous ranked probability score (CRPS), which evaluates the quality of the entire predictive distribution, including its shape and spread \citep[see, e.g.,][]{Gneiting2014}. For the heavy-tailed observations, we compared the state filtering distribution for the different methods to the true state, whereas for rainfall data we considered predictions of the true rainfall, $x_{t,l}\textcolor{black}{\mathbbm{1}_{\{x_{t,l}>0\}}}$.

The GEnKF outperformed all competitors in all three scenarios, and came relatively close to the true posterior distribution. This is especially remarkable considering the GEnKF consisted of at most three EnKF updates with ensemble size 30, while the competitors each used an ensemble of size 100. Hence, in our simulations, the GEnKF produced the most accurate results at the lowest computational expense.

{\color{black}
\begin{figure}
	\begin{subfigure}{.48\textwidth}
	\centering
	\includegraphics[width =1\linewidth]{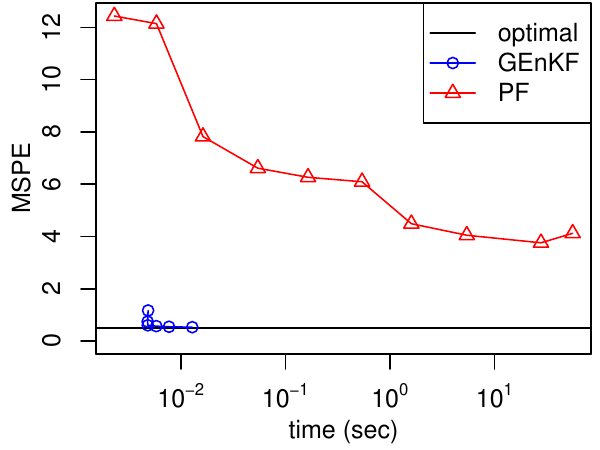} 
	\caption{$n=100$, $m=75$}
	\label{fig:timingn100}
	\end{subfigure}%
\hfill
	\begin{subfigure}{.48\textwidth}
	\centering
	\includegraphics[width =1\linewidth]{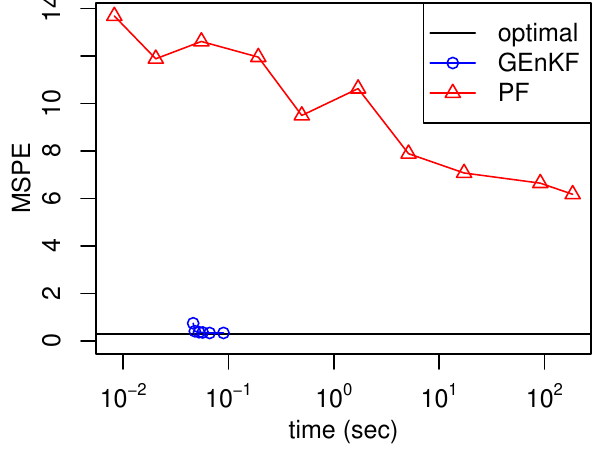}
	\caption{$n=400$, $m=300$}
	\label{fig:timingn400}
	\end{subfigure}%
\caption{{\color{black}Mean squared prediction error (MSPE) versus time for GEnKF and particle filter (PF) updates in the simulated rain example for known $\kappa=3$. Note that time on the x-axis is on a log scale.}}
\label{fig:timing}
\end{figure}

In Figure \ref{fig:timing}, we further investigated the computational costs by comparing the computation time and MSPE for the update in the GEnKF and the particle filter for varying ensemble sizes between 5 and 100 and number of particles between 30 and 1 million. The results ignore the potentially significant cost of applying evolution operators, to the benefit of the particle filter. All results are averaged over 20 repetitions. While the GEnKF obtains close to optimal MSPE even for very low computation times, the particle filter produces highly non-optimal MSPEs even for computation times that are several orders of magnitude larger. The differences become even more pronounced for larger datasets.

}

\subsection{Cloud data \label{sec:penkfexample}}

We also conducted filtering inference for a cloud-motion dataset shown in Figure \ref{fig:clouddata} from \citet{Wikl:02}, representing cloud intensities (i.e., counts) at $n=60$ locations along a spatial transect at $T=80$ time points. The data roughly follow an overdispersed Poisson distribution, \textcolor{black}{and so we assumed the model from Example \ref{ex:poisson}.} We randomly set aside 10\% of the observations as test data, so that $m_t \equiv 54$ and $\bH_t$ was a subset of the identity. \textcolor{black}{The observation locations were different for each time period.} The evolution was taken to be linear, $\evol_t(\bx_{t-1}) = \levol_t\bx_{t-1}$, where $\levol_t$ was a tridiagonal matrix with $\gamma_{1,t}$ on the main diagonal, and $\gamma_{2,t}$ and $\gamma_{3,t}$ on the first diagonals above and below, respectively. The evolution error was assumed to exhibit spatial dependence in the form of a Mat\'{e}rn covariance with smoothness 1.5, such that the $(i,j)$th element of $\bQ_t$ is given by $(\bQ_t)_{i,j} = \tau^2_t \,\matern(|i-j|/\lambda_t)$, where $\matern(d) = (1+\sqrt{3}d)\exp(-\sqrt{3}d)$. The unknown parameters were $\bftheta_t \colonequals (\gamma_{1,t},\gamma_{2,t},\gamma_{3,t},\log \sigma_t,\log \tau_t,\log \lambda_t)'$. For $t=0$, we set $\bftheta_0 = (0.3,0.3,0.3,\log(0.1),\log(1.5),\log(8))'$, $\bfmu_{0|0} = (-2) \bfone_n$, and $(\bfSigma_{0|0})_{i,j} = 0.2 \, \matern(|i-j|/5)$. Note that our model addresses the potential extension to time-varying parameters suggested in \citet{Wikl:02}.

\begin{figure}
	\begin{subfigure}{.47\textwidth}
	\centering
	\includegraphics[trim=0mm 0mm 10mm 0mm, clip, width =1\linewidth]{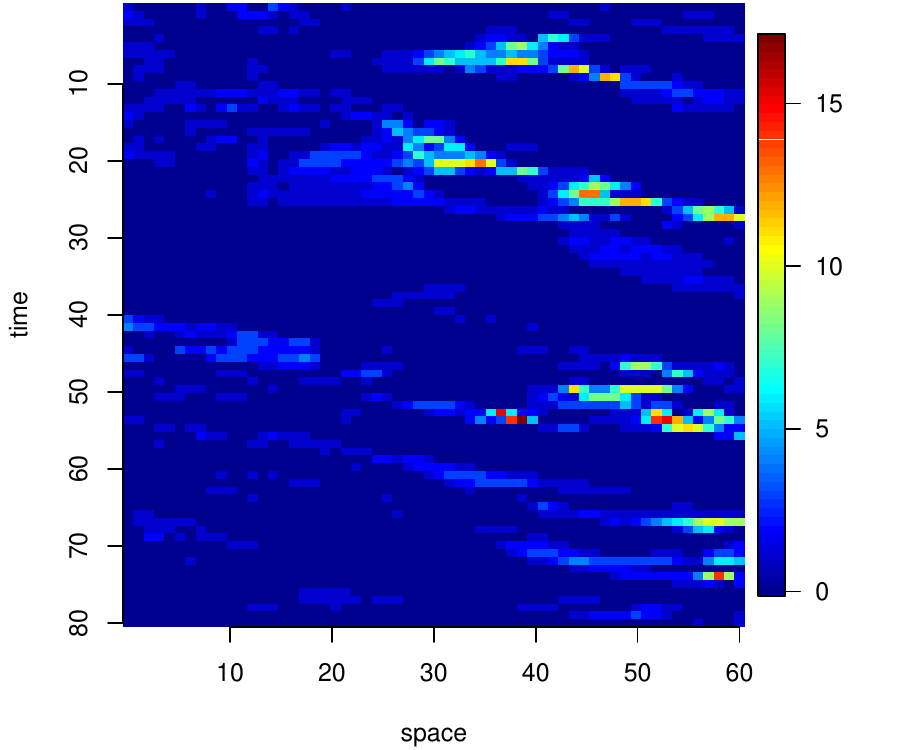} 
	\caption{Cloud intensity data}
	\label{fig:clouddata}
	\end{subfigure}%
\hfill
	\begin{subfigure}{.5\textwidth}
	\centering
	\includegraphics[width =1\linewidth]{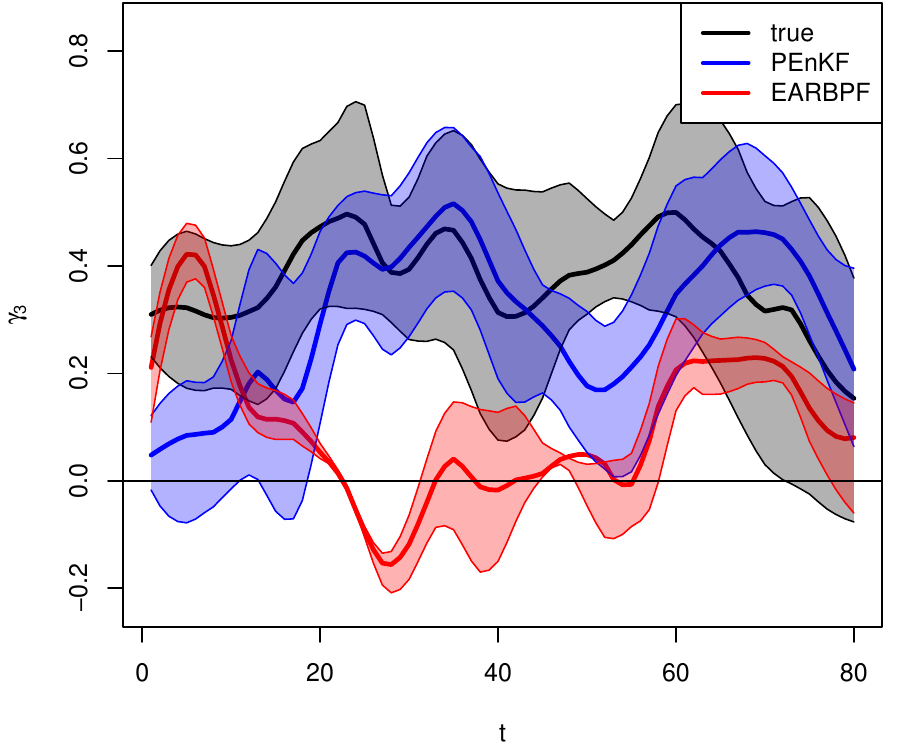}
	\caption{Filtering distributions for $\gamma_{3,t}$}
	\label{fig:cloudgamma}
	\end{subfigure}%
\caption{The cloud data in Section \ref{sec:penkfexample} (Panel (\subref{fig:clouddata})), and filtering means and 95\% credible intervals (slightly smoothed for clearness) over time for $\gamma_{3,t}$ (Panel (\subref{fig:cloudgamma}))}
\label{fig:cloud}
\end{figure}

We considered two methods for filtering inference on the state and parameters: our PEnKF from {\color{black}Example \ref{ex:poissonpenkf}}, and an exact approximation of a Rao-Blackwellized particle filter \citep[EARBPF;][]{Johansen2012}. Forecast independence does not hold for all parameters in this example, and so we \textcolor{black}{used} $M=50$ ensembles of size $N=30$. For tapering of the forecast covariance matrix (see Section \ref{sec:tapering}), we chose a Wendland taper with range 8. For the EARBPF, which is similar to the PEnKF except that it uses local particle filters instead of the EnKF and the Laplace approximation to integrate out the state for each parameter particle, we considered one setting with a similar computation time as the PEnKF ($M=50$ and $N=30$), and one setting with a much larger computation time ($M=500$ and $N=10000$) that should provide a close approximation to and will henceforth be referred to as the true filtering distribution.

\begin{table}[ht]
\centering
\begin{tabular}{r|rr|rrrrrr}
 & \multicolumn{2}{c|}{Predictions} & \multicolumn{6}{c}{EMD to true filtering distributions}  \\ 
 & MSPE & CRPS & $\gamma_1$ & $\gamma_2$ & $\gamma_3$& $\log \sigma$ & $\log \tau$ & $\log \lambda$ \\ 
  \hline
truth & 0.73 & 0.25 &  &  &  &  &  &  \\ 
	  \textbf{PEnKF} & \textbf{0.75} & \textbf{0.25} & \textbf{0.25} & \textbf{0.16} & \textbf{0.19} & \textbf{0.70} & \textbf{0.23} & \textbf{0.32} \\ 
  EARBPF & 1.26 & 0.33 & 0.27 & 0.25 & 0.24 & 0.71 & 0.32 & 0.37 \\ 
\end{tabular}
\caption{For the cloud data in Section \ref{sec:penkfexample}, mean squared prediction error (MSPE) and continuous ranked probability score (CRPS) for prediction of test data, and earth-mover's distance (EMD) to the true parameter posteriors}
\label{tab:cloud}
\end{table}

The results, averaged over ten different initial ensembles $\bx_{0|0}^{(1:M,1:N)}$, are shown in Table \ref{tab:cloud}. In terms of filtering predictions of the test data, the PEnKF performed almost as well as the true filtering distributions, and considerably better than the EARBPF. The PEnKF approximation to the filtering distributions of the six parameters in $\bftheta_t$ is also closer than the EARBPF to the true distributions in terms of earth-mover's distance \citep[computed using the R package by][]{emdist} averaged over time. As an example, Figure \ref{fig:cloudgamma} shows the filtering distributions for one initial ensemble over time of the parameter $\gamma_{3,t}$, which can be viewed as quantifying the amount of movement of the cloud intensities ``one pixel to the right'' from one time point to the next. This rightward drift, which is clearly present in the data (see Figure \ref{fig:clouddata}), is picked up fairly quickly by the true and PEnKF filtering distributions, while the EARBPF approximation even implies a negative $\gamma_{3,t}$ between $t=25$ and $t=30$.


\subsection{Smoothing inference for the Lorenz-96 model \label{sec:genksex}}

Finally, we considered smoothing inference for simulated data based on the popular nonlinear Lorenz-96 model \citep{Lore:96} 
{\color{black}
described in Example \ref{ex:lorenz}, with $\bQ = 0.2\bfSigma_L$, where $\bfSigma_L$ is the covariance matrix obtained by a long run of the $\lor_{8,0.2}$ model. Starting with $\bfmu_{1|0} = \bfzero$ and $\bfSigma_{1|0} = \bfSigma_L$, we simulated data from the model 
for $t=1,\ldots,T$ with $T=10$, assuming $\theta \sim \normal(\mu_\theta,\sigma_\theta^2)$ with $\mu_\theta=0.8$ and $\sigma_\theta=0.2$.
The Lorenz-96 equations were solved numerically using an Euler scheme.}

We simulated 100 such datasets, with the true $\theta$ drawn for each simulation from the prior. For each dataset, we ran three different methods to obtain approximations to the smoothing distribution \textcolor{black}{$p(\theta,\bx_{1:T}|\by_{1:T})$} using a starting value of $\theta=0.5$. The first method considered was the EnKS in Algorithm \ref{alg:enks} with $N= $ 1,000 and state augmentation; that is, $\theta$ was treated as part of the state vector $\bx_t$ with artificial evolution distribution $\theta_t | \theta_{t-1} \sim \normal(\theta_{t-1},0.1^2)$. We also considered the GEnKS in \textcolor{black}{Example \ref{ex:lorenzgenks}}. For both EnKS algorithms with used a taper over time and space given by a Wendland correlation function with radius 3 and 8, respectively. Finally, we considered the particle Gibbs sampler \citep[][Sect.~2.4.3]{Andrieu2010}, \textcolor{black}{which uses the same distribution as the EnKS in \eqref{eq:lorenzfcd} to update $\theta$, but uses a ``conditional'' particle filter to sample from the FCD of the state $\bx_{1:T}$ given the data and $\theta$.} For both the GEnKS and the particle Gibbs sampler, we used an ensemble size of $N=50$ and 100 Gibbs iterations, the first 20 of which were considered burn-in.

Due to the nonlinear evolution, it is not possible to obtain the true smoothing distribution of the states and parameter in any reasonable computation time. Hence, we compared the three methods using the same proper scoring rules as in Section \ref{sec:nongaussobs}, averaged over the 100 simulated datasets. The results are summarized in Table \ref{tab:lorenz}. Our GEnKS algorithm performed best both in terms of inference on the parameter and on the states. The state augmentation method failed, leading to very high (i.e., bad) scores for the posterior on $\theta$. The particle Gibbs sampler produced worse inference on $\theta$ than simply using the prior distribution (i.e., completely ignoring any information in the data). Some of the results are also summarized in Figure \ref{fig:lorenz}, which shows the \emph{collapse} of the particle Gibbs procedure for inference on $\bx_{1:T}$ and its non-convergence in terms of inference on $\theta$. 

\begin{table}[ht]
\centering
\begin{tabular}{r|rr|rr}
 & \multicolumn{2}{c|}{Parameter $\theta$} & \multicolumn{2}{c}{State $\bx_{1:T}$}  \\ 
 & MSPE & CRPS & MSPE & CRPS \\ 
  \hline
\textbf{GEnKS} & \textbf{0.002} & \textbf{0.024} & \textbf{0.710} & \textbf{0.478} \\ 
  EnKS+SA & $>$100 & $>$100 & 0.914 & 0.540 \\ 
  Particle Gibbs & 0.111 & 0.262 & 12.380 & 2.495 \\ 
  Prior & 0.042 & 0.118 &  &  \\ 
\end{tabular}
\caption{\label{tab:lorenz}Mean squared prediction error (MSPE) and continuous ranked probability score (CRPS) for smoothing inference on the parameter $\theta$ and states $\bx_{1:T}$ for the Lorenz-96 simulations in Section \ref{sec:genksex}}
\end{table}

Note that while we considered $\theta$ to be the only unknown parameter here for clearness of exposition, our GEnKS should also work well for unknown parameters in $\bH_t$ and scalar multiplicative parameters in $\bQ_t$ and $\bR_t$, in that closed-form FCDs are available in both of these cases if a normal and inverse-gamma priors, respectively, are assumed.

\begin{figure}
	\begin{subfigure}{.48\textwidth}
	\centering
	\includegraphics[width =1\linewidth]{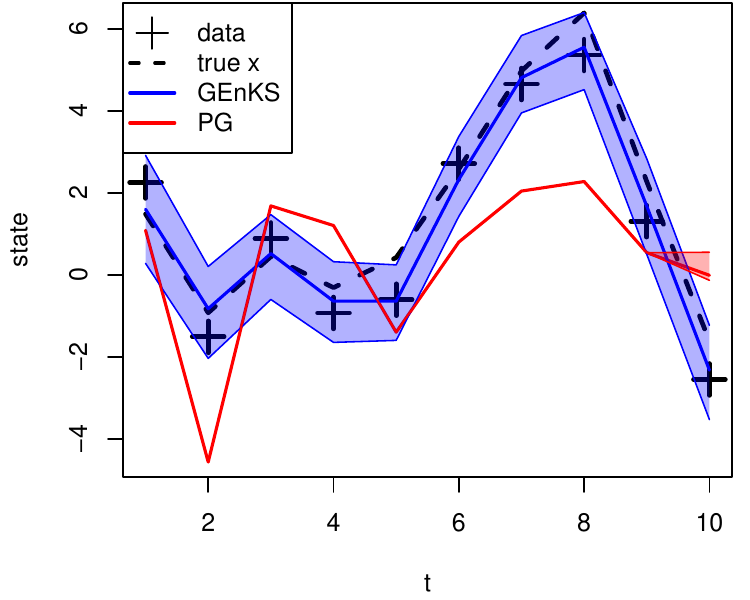}
	\caption{State at location 1 over time (i.e., $\bx_{1:T,1}$)}
	\label{fig:lorloc}
	\end{subfigure}%
\hfill
	\begin{subfigure}{.48\textwidth}
	\centering
	\includegraphics[width =1\linewidth]{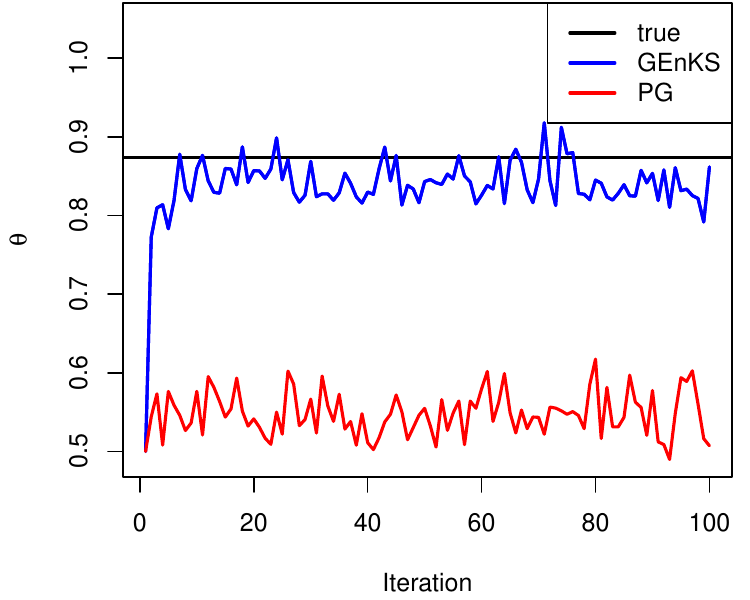}
	\caption{Trace plots for $\theta$}
	\label{fig:lortrace}
	\end{subfigure}%
\caption{For one of the Lorenz-96 simulations in Section \ref{sec:genksex}, true values and posterior distributions (i.e., posterior means and pointwise 80\% prediction intervals) of the state at a single location (Panel (\subref{fig:lorloc})), and trace plots for $\theta$ (Panel (\subref{fig:lortrace}))}
\label{fig:lorenz}
\end{figure}

\section{\textcolor{black}{Conclusions} \label{sec:conclusions}}

\textcolor{black}{We have introduced a new class of ensemble filtering and smoothing algorithms for Bayesian inference in high-dimensional dynamic spatio-temporal models.  The algorithms combine existing approaches from the state-space literature with ensemble Kalman methods from geophysics.   The main advantage of our algorithms is that they scale well to high-dimensional models, and they allow for nonlinearity, non-Gaussian observations, and unknown parameters.  In contrast, most existing state-space algorithms break down in high dimensions, and most ensemble Kalman approaches are poorly suited for non-Gaussian models with unknown parameters. We illustrated the proposed methods with four examples: two non-Gaussian models, a nonlinear geophysical system, and a real data example of cloud motion. We showed that our methods outperform the main competing approaches, including EnKFs, particle filters, and particle MCMC. Our methods also provide an approach for approximate Bayesian inference for large spatial data, in the special case of only $T=1$ time point.}

\textcolor{black}{
We have proposed two types of algorithms: {\em Ensemble Gibbs} and {\em Ensemble Marginalization} methods.  For Ensemble Gibbs methods, the EnKF is used for approximate simulation from the full conditional state distribution.  In our examples, we show that these approximate Gibbs samplers can provide accurate state and parameter inference and converge within a small number of iterations.  For Ensemble Marginalization methods, the EnKF is used to approximately integrate out the states and simulate from their full conditional posterior.  In this context, we have studied the theoretical properties of EnKF-based likelihood estimates and shown that an ensemble size that is linear in the effective dimension is sufficient to bound the variance of the estimated loglikelihood.  This is in contrast to particle-filter-based likelihood estimates, for which the required number of particles grows exponentially in the effective dimension. We leave as future research further study on the theoretical properties of the proposed algorithms.}

\textcolor{black}{
Statisticians now acknowledge that fast, approximate methods are necessary to analyze very large datasets.  The EnKF, developed in the geophysics community, has been successfully applied for fast, approximate inference on models with millions of state variables.   We believe that combining standard statistical approaches with these approximate but scalable EnKF methods provides solutions to problems that might be otherwise unsolvable.   }

\small
\appendix

\section*{Acknowledgments}

Katzfuss' research was partially supported by National Science Foundation (NSF) Grant DMS--1521676 and NSF CAREER Grant DMS--1654083. Wikle acknowledges the support of NSF grant SES-1132031, funded through the National Science Foundation Census Research Network program. We would like to thank Fuqing Zhang, Amir Nikooienejad, \textcolor{black}{the Associate Editor, and two anonymous referees for helpful comments and suggestions.}

\section{Proofs \label{app:proofs}}

\begin{proof}[Proof of Proposition \ref{prop:pflik}]
For notational simplicity, we drop the $\widetilde~$ on $\forex^{(j)}$, \textcolor{black}{and we assume $\bfmu =\bfzero$ without loss of generality}. First, note that $\sum_{i=1}^n (x_i^{(j)}-y_i)^2$ follows a noncentral $\chi^2$-distribution, and so
\[
\textstyle\log \plik(\by|\bftheta,\bx^{(j)})=-\frac{n}{2}\log(2\pi\theta) - \frac{1}{2\theta}\sum_{i=1}^n (x_i^{(j)}-y_i)^2 \stackrel{d}{\rightarrow}\normal(na,nb),
\]
where $a=-\frac{1}{2}(\log(2\pi\theta)+\kappa+\frac{1}{n}\sum_{i=1}^n y_i^2)$ and $b=2\kappa(\kappa+\frac{2}{n}\sum_{i=1}^n y_i^2)$. Using properties of the lognormal distribution, we have
\[
\textstyle \plik(\by|\bftheta,\bx^{(1:N)})= \frac{1}{N} \sum_{j=1}^N \exp( \log \plik(\by|\bftheta,\bx^{(j)})) \stackrel{d}{\rightarrow} \normal(e^{n(a+b/2)},\frac{1}{N}e^{n2(a+b)}),
\]
and finally, using the delta-method, we have the asymptotic variance
\[
 \textstyle\var \log \plik(\by|\bftheta,\bx^{(1:N)}) \stackrel{a}{=} \frac{1}{N}e^{n2(a+b)}e^{-2n(a+b/2)} = \frac{1}{N}e^{nb}=\frac{1}{N}e^{n2\kappa(\kappa+(2/n)\sum_{i=1}^n y_i^2)} = \mathcal{O}(e^n/N).
\]
\end{proof}

\begin{proof}[Proof of Proposition \ref{prop:enkflik}]
\textcolor{black}{We drop the $\widetilde~$ on $\forex^{(j)}$, and assume $\bfmu =\bfzero$ without loss of generality.} Due to the diagonal tapering, the covariance matrix in the likelihood is a diagonal matrix with $i$th element 
\[
\textstyle\hat\sigma^2_i = \frac{1}{N-1} \sum_{j=1}^N (x_i^{(j)} - \hat\mu_i)^2 + \theta \; \stackrel{d}{\rightarrow}\; \normal(\sigma^2,2\kappa^2/N),
\]
where $\sigma^2 = \kappa+\theta$ and $\hat\mu_i = \frac{1}{N}\sum_{j=1}^N x_i^{(j)}$. Because $\var(\hat\mu_i^2) = \mathcal{O}(1/N^2)$, the variance of the mean estimator is negligible relative to that of $\hat\sigma^2_i$, and we regard $c_i = (y_i - \hat\mu_i)^2$ as a constant. Thus, by the delta method,
\[
\textstyle\var(\log\elik(y_i|\theta,x_i^{(1:N)})) =\var(\frac{1}{2}\log(\hat\sigma^2_i) + \frac{1}{2}\frac{c_i}{\hat\sigma^2_i}) \stackrel{a}{=} \frac{1}{4} \frac{2\kappa^2}{N}(\frac{1}{\sigma^2} - \frac{c_i}{\sigma^4})^2.
\]
Finally, we have $\log\elik(\by|\theta,\bx^{(1:N)}) = \sum_{i=1}^n \log\elik(y_i|\theta,x_i^{(1:N)})$, and so $\var(\log\elik(\by|\theta,\bx^{(1:N)})) = \mathcal{O}(n/N)$.
\end{proof}

{\color{black}
\begin{proof}[Proof of Proposition \ref{prop:convergence}]
For linear Gaussian SSMs, the EnKF converges to the Kalman filter as $N \rightarrow \infty$ \citep[e.g.,][]{Butala2008,Mandel2011}, as long as $\widehat{\bfSigma}_{t|t-1}$ is a consistent estimator of the forecast covariance matrix. As the EnKS can be viewed as an EnKF on an extended state-space, the same result holds for the EnKS. As stated in Section \ref{sec:elik}, the EnKF likelihood $\elik_t$ also converges to the true likelihood at each time $t$. Combining this with convergence results for MCMC \citep[e.g.,][]{Robert2004} and SMC \citep[e.g.,][]{Doucet2001} methods under mild regularity conditions, it is clear that the PEnKF, GEnKS, MHEnKS, and PEnKS all converge to the true filtering or smoothing distributions in a HSSM with linear evolution operators $\evol_t$, as $N \rightarrow \infty$ and $M$ or the number of MCMC iterations increase. 
\end{proof}
}

\footnotesize
\bibliographystyle{apalike}
\bibliography{../enkfbib}

\end{document}